\documentclass[review,onefignum,onetabnum]{siamart171218}

\usepackage{tikz}
\usepackage{tikz-network}
\usepackage{amsmath}
\usepackage{amssymb}
\usepackage{amsfonts}
\usepackage{graphicx}
\usepackage{enumitem}
\usepackage{mathrsfs}
\usepackage{float}
\usepackage{cite}
\usepackage{tikz}
\usetikzlibrary{matrix,arrows,decorations.pathmorphing}
\usetikzlibrary{decorations.markings}
\usepackage{verbatim}
\usepackage{mathtools}
\usepackage{caption}
\usepackage{subcaption}
\usetikzlibrary{arrows}
\usepackage{tabularx}
\usepackage{xcolor}
\usepackage{inputenc}
\usepackage[T1]{fontenc}
\newtheorem{pf}{Proof}
\newtheorem{prop}{Proposition}
\newtheorem{remark}{Remark}
\newtheorem{assumption}{Assumption}

\DeclareMathOperator*{\Dg}{\mathsf{diag}}

\renewenvironment{pf}{{ {\textbf{Proof}}$\;$}}{\par}
\usepackage{lipsum}
\usepackage{amsfonts}
\usepackage{graphicx}
\usepackage{epstopdf}
\usepackage{algorithmic}
\ifpdf
  \DeclareGraphicsExtensions{.eps,.pdf,.png,.jpg}
\else
  \DeclareGraphicsExtensions{.eps}
\fi

\newsiamremark{hypothesis}{Hypothesis}
\crefname{hypothesis}{Hypothesis}{Hypotheses}
\newsiamthm{claim}{Claim}

\headers{General SIS on Hypergraphs}{S. Cui, F. Liu, H. Jardón-Kojakhmetov, and M. Cao}

\title{General SIS diffusion process with indirect spreading pathways on a hypergraph
}

\author{Shaoxuan Cui\thanks{Bernoulli Institute, Faculty of Science and Engineering, University of Groningen, The Netherlands 
  (\email{s.cui@rug.nl}).}
\and Fangzhou Liu\thanks{Research Institute of Intelligent Control and Systems, School of Astronautics, Harbin Institute of Technology, China 
  (\email{fangzhouliu.ac@gmail.com}).}
\and Hildeberto Jard{\'o}n-Kojakhmetov \thanks{Bernoulli Institute, Faculty of Science and Engineering, University of Groningen, The Netherlands (\email{h.jardon.kojakhmetov@rug.nl}).}
\and
Ming Cao \thanks{ENTEG, Faculty of Science and Engineering, University of Groningen, The Netherlands (\email{m.cao@rug.nl}).}
}

\usepackage{amsopn}

\makeatletter
\newcommand*{\addFileDependency}[1]{
  \typeout{(#1)}
  \@addtofilelist{#1}
  \IfFileExists{#1}{}{\typeout{No file #1.}}
}
\makeatother

\newcommand*{\myexternaldocument}[1]{%
    \externaldocument{#1}%
    \addFileDependency{#1.tex}%
    \addFileDependency{#1.aux}%
}

\ifpdf
\hypersetup{
  pdftitle={General SIS diffusion process with indirect spreading on a hypergraph},
  pdfauthor={S. Cui, F. Liu, H. Jardon-Kojakhmetov, and M. Cao}
}
\fi

\myexternaldocument{ex_supplement}

\begin{document}
\nolinenumbers
\maketitle

\begin{abstract}
Conventional graphs captures pairwise interactions; by contrast, higher-order networks (hypergraphs, simplicial complexes) are introduced to the interactions involving more parties, which has been rapidly applied to characterize a growing number of complex real systems. It becomes apparent that higher-order interactions give rise to new challenges for the rigorous analysis of dynamics on a higher-order network.
In this paper, we study a series of Susceptible-infected-susceptible-type (SIS-type) diffusion processes with both indirect and direct pathways on a directed hypergraph. We start with a concrete setup where we choose a polynomial interaction function to describe how several agents influence one another over a hyperedge. Then, under this interaction function, we further extend the system and propose a bi-virus competing model on a directed hypergraph by coupling two single-virus models together. Finally, the most general model in this paper considers a general interaction function for both the single-virus and bi-virus cases. For the single-virus model, we provide a comprehensive characterization of the healthy state and endemic equilibrium. For the bi-virus setting, we further give the analysis of the existence and stability of the healthy state, dominant endemic equilibria, and coexisting equilibria. All theoretical results are supported additionally by some numerical examples.

\end{abstract}

\begin{keywords}
  Susceptible-infected-susceptible (SIS) model, Vector transmission, Networked systems, Hypergraphs, Dynamical systems
\end{keywords}

\begin{AMS}
  05C65, 34D05, 34C12, 37N25, 92D30
\end{AMS}

\section{Introduction}
Mathematical modeling of epidemic spreading processes is crucial in epidemiology and has a long history of scientific study. Especially during the recent Covid-19 period, epidemic models have demonstrated even greater importance \cite{vespignani2020modelling,stella2022role}. The precise and timely  prediction \cite{achterberg2022comparing}, parameter identification \cite{sauer2021identifiability,ye2021parameter} and design of control strategies \cite{kohler2021robust
,carli2020model,kovacevic2022distributed,yi2022edge} using epidemic models have provided us with strong guidance to fight against the pandemic. Besides the scenario of the spreading of a virus, epidemic models also show great potential to describe very general diffusion processes, including information diffusion in social networks~\cite{li2017survey,zhu2014demographic,8954786}, computer virus dissemination \cite{lloyd2001viruses,kephart1992directed,yang2012towards} and others \cite{illner2016sis,rodrigues2016can}. 

Among various epidemic models, compartmental models are especially popular \cite{mei2017dynamics,zino2021analysis,pare2020modeling}. In such a framework, each compartment indicates a state, and transmissions between compartments show changes between the states of infections. For example, usually, $S$ denotes the susceptible, $I$ is the infected, and $R$ is the recovered states. By considering various combinations of state transitions, epidemic models, such as SIS, SIR, and SIRS, have been proposed. In the early stage of the development of compartmental epidemic models, a relatively simple modeling framework, called \emph{scalar models}, was proposed and studied. In this framework, a single population is studied and is assumed to be well-mixed, i.e., everyone within the population has the same chance to get infected. This simple model is detailed in, e.g., the survey \cite{hethcote2000mathematics}.

It is evident that such a simple framework with its assumption on homogeneity provides only a relatively rough approximation of the real complex diffusion process and fails to capture the heterogeneity in the real world (for example, age structures, spatial diversity, and social behaviors). In order to overcome this drawback, the population can be divided into several sub-populations, and a networked structure (hence referred to as \emph{network models}) is further introduced to describe the interactions among the sub-populations. For instance, the contributions to the networked SIS models include those in continuous-time \cite{van2008virus,van2013homogeneous,khanafer2014stability} and others in discrete-time \cite{pare2018analysis,prasse2019viral,liu2020stability}. By further considering that two or more viruses spread over the network, the bi- or multi-virus networked SIS models have also been considered, for example in \cite{ye2022convergence,liu2019analysis,gracy2022endemic,pare2021multi,pare2020analysis}. All these above-mentioned models assume that the interaction between sub-populations is linear and thus additive. Such an assumption may be violated under some situations since the infection process sometimes may have the saturation effect and thus slow down \cite{liu1987dynamical,yang2015impact,yuan2012modeling,ruan2003dynamical}. With the introduction of non-linear interaction functions, the networked bi-virus SIS model \cite{yang2017bi} is able to capture the saturation effect of the infection process.

\emph{Networked epidemic models} still use the conventional graphs, one edge of which just connects one node to another. This leads to the restriction that dynamical systems on conventional graphs only take pairwise interactions into account. In contrast, it can be frequently observed that people not only interact with each other one one one but also as a group. Recently, it has been shown that higher-order networks, including \emph{hypergraphs} and \emph{simplicial complexes}, provide a powerful modeling tool to characterize effects beyond pairwise interactions \cite{bick2021higher}. In this higher-order network structure, one hyperedge can contain more than two nodes and it is thus suitable to model group-wise multi-body interactions. Along this line, an SIS model on a simplicial complex is proposed and further studied via a mean-field approximation in \cite{iacopini2019simplicial}. In each simplex, the healthy agent can only get infected via the hyperedge when all other agents in the hyperedge are infected. Followed by this modeling framework and by further modifications, a similar SIS model on a hypergraph is proposed and studied in \cite{cisneros2021multi}. By calculating the probabilities of infection and using epidemic link equations \cite{matamalas2018effective}, another type of SIS models on a simplicial complex is studied in \cite{matamalas2020abrupt}. In \cite{de2020social}, an SIS model on a hypergraph is proposed, and most analytical results are achieved via adopting a special network structure (hyperstar or homogeneous hypergraph). Furthermore, the model they studied is a threshold model and one agent may get infected through the hyperedge when more than a certain amount of the agents in the hyperedge are infected. By considering that one may react differently in different surroundings, an appropriate model is further proposed and studied in \cite{st2021universal}. In \cite{li2022competing}, a bi-virus SIS model on a simplicial complex has been studied. However, under the mean-field approximation and homogeneous assumption, the model in \cite{li2022competing} is approximated by a planar system, and the coexistence results are established purely via simulations, which restricts the model from wider applications. 

Apart from discrepancies in network topology, all the aforementioned models only describe direct human-to-human spreading processes. However, in a real epidemic process, especially those concerning waterborne or insect-borne viruses, the \emph{pathogen} also spreads through different kinds of media. Analogously, for information diffusion on social networks, a post on Twitter or a video on Tiktok both facilitates information spreading. For this reason, it is also of great significance to capture the indirect spreading process through different types of media. This is usually realized by coupling the aforementioned epidemic models with extra pathogen dynamics on resources (or media). With this idea, the SIWS\footnote{Here, the $W$ stands for ``water'' but it is only indicative of the medium, and is by no means restricted to it.}\footnote{In this paper, SIWS means susceptible-infected-water-susceptible, and the use of $W$ is different from waning immunity, which has also been studied in \cite{chaves2007loss,lavine2011natural,jardon2021geometric}.} model on a graph, which is an SIS model coupled with pathogen dynamics, is proposed and studied by, for example, \cite{gracy2022analysis,janson2020networked,janson2020analysis,pare2022multi,liu2019networked,cui2022discrete}. However, the analysis of SIWS models on higher-order networks is, although extremely relevant, still missing.

The main purpose of this paper is to propose and study a general framework for spreading processes on hypergraphs. It is then very natural to think about an SIWS model on a higher-order network potentially with non-linear and non-polynomial interaction functions. 


Briefly, we list the contributions and organization of our paper as follows.
Firstly, in section \ref{sec:model}, making use of a specific choice of interaction function, we propose a single-virus SIWS model on a directed hypergraph, which incorporate both a direct human-to-human pathway and an indirect human-resource-human pathway. For the modeling, we introduce interaction functions, the spreading mechanism of different orders, and higher-order network structures (hypergraphs, simplicial complex). 

In section \ref{sec::main}, we give a full analysis of the model proposed in section \ref{sec:model}. We give the results regarding the healthy state and the endemic equilibrium. Especially, we find out that bistability may occur due to the introduction of higher-order interactions related to the hypergraph structure. We also confirm that the proposed model is an irreducible monotone system and presents no chaos.

In section \ref{sec:bi}, we further consider a bi-virus competing scenario and provide abundant analytical results. Particularly, we give the results regarding the healthy state, dominant endemic equilibrium, and coexisting equilibrium. The bi-virus system is also an irreducible monotone system and presents no chaos.

In section \ref{sec:abstract}, to further generalize our framework by considering an abstract interaction function and prove similar results, for both the single-virus and the bi-virus cases. This general system remains irreducible monotone.

Finally, in section \ref{sec:sim}, some numerical examples are given to highlight the theoretical contributions.

\section{Preliminaries}
{A tensor $T\in\mathbb{R}^{n_1\times n_2 \times \cdots \times n_k}$ is a multidimensional array, where the order is the number of its dimension $k$ and each dimension $n_i$, $i=1,\cdots,k$ is a mode of the tensor. A tensor is cubical if every mode has the same size, that is $T\in\mathbb{R}^{n\times n \times \cdots \times n}$. We further write $k$-th order cubical tensor as $T\in\mathbb{R}^{n\times n \times \cdots \times n}=\mathbb{R}^{[n,k]}$. A cubical tensor $T$ is called supersymmetric if $T_{j_1 j_2 \ldots j_k}$ is invariant under any permutation of the indices.}

The definition of a hypergraph that we consider in this paper follows from \cite{gallo1993directed}. Furthermore, we introduce a set of several tensors to collect the information of the weights. A weighted and directed hypergraph is a triplet $\mathbf{H}=(\mathcal{V},\mathcal{E}, A)$. The set $\mathcal{V}$ denotes a set of vertices and $\mathcal{E}=\{E_1, E_2, \cdots \}$ is the set of hyperedges. A hyperedge is an ordered pair $E=(\mathcal{X},\mathcal{Y})$ of disjoint subsets of vertices; $\mathcal{X}$ is the tail of $E$ while $\mathcal{Y}$ is the head. According to the physical meaning of the scenario of a diffusion process, we assume that each hyperedge has only one tail but one or multiple ($\geq 1$) heads. The set of tensors $A=\{A_2, A_3,\cdots\}$ collect the weights of all hyperedges, where $A_2=[A_{ij}]$ collects the weights of all second-order hyperedges, $A_3=[A_{ijk}]$ collects the weights of all third-order hyperedges, and so on. For instance, $A_{ijkl}$ denotes the weight of the hyperedge where $i$ is the tail and $j,k,l$ are the heads. For simplicity, in this paper, we also use the weight (for example, $A_{\bullet}$) to denote the corresponding hyperedge. If all hyperedges only have one tail and one head, then the network is a standard directed and weighted graph. As a special case, a weighted and undirected hypergraph is a triplet $\mathbf{H}=(\mathcal{V},\mathcal{E}, A)$, where $\mathcal{E}$ is now a finite collection of non-empty subsets of $\mathcal{V}$ \cite{bick2021higher}. {An undirected hypergraph can be represented as a supersymmetric tensor. For simplicity, we can use the notation $\pi(\bullet)$ , where $\pi(\bullet)$ denotes any permutation of $\bullet$, to indicate the undirected hyperedge.  For example, an undirected hyperedge can be written as $A_{\pi(ijk)}$ since $A_{ijk},A_{ikj},A_{kij},A_{kji},A_{jik},A_{jki}$ take the same value in an undirected graph.} An undirected simplicial complex is a special type of undirected hypergraph that contains all nonempty subsets of hyperedges as hyperedges.  If $A_{\pi(ijk)}$ is a hyperedge of the simplicial complex then $A_{\pi(ij)}, A_{\pi(ik)},A_{\pi(jk)}$ are also hyperedges of the simplicial complex. In order to capture the heterogeneity of the spreading process, we extend the definition of the simplicial complex into a directed simplicial complex. That is, a simplicial complex is a special type of directed hypergraph such that 
 for all $E=(\mathcal{X},\mathcal{Y})\in \mathcal{E}$ and $\tilde{E}=(\mathcal{\tilde{X}},\mathcal{\tilde{Y}}),\mathcal{\tilde{X}}\subset\mathcal{{X}},\mathcal{\tilde{Y}}\subset\mathcal{{Y}}$, then $\tilde{E}\in \mathcal{E}$;
For example, if $A_{ijk}$ is a hyperedge of the directed simplicial complex then $A_{ij}, A_{ik}$ are also hyperedges of the simplicial complex.

\emph{Notation:} $\mathbb{R}$ and $\mathbb{N}$ denote the set of real numbers and nonnegative integers, respectively. Given a square matrix $M \in \mathbb{R}^{n \times n}$, $\rho(M)$ is the spectral radius of $M$, which is the largest absolute value of the eigenvalues of $M$. By $s(M)$, we refer to the the largest real part among the eigenvalues of
$M$. For a matrix $M \in \mathbb{R}^{n \times r}$ and a vector $a \in \mathbb{R}^n$, $M_{ij}$ and $a_{i}$ denote the element in the $i$th row and $j$th column and the $i$th entry, respectively. For any two vectors $a, b \in \mathbb{R}^n$, $a \gg (\ll) b$ represents that $a_i >(<) b_i$, for all $i=1,\ldots,n$; $a > (<) b$ means that $a_i \geq (\leq) b_i$, for all $i=1,\ldots,n$ and $a \neq b$; and $a \geq (\leq) b$ means that $a_i \geq (\leq) b_i$, for all $i=1,\ldots,n$ or $a = b$. These component-wise comparisons are also applicable to matrices with the same dimension. The vector $\mathbf{1}$ ($\mathbf{0}$) represents the column vector or matrix of all ones (zeros) with appropriate dimensions. The matrix $I$ represents the identity matrix with the appropriate dimension.
\section{Model Set-up}\label{sec:model}
In this section, we introduce our model set-up and some important concepts (spreading mechanisms, interaction functions, indirect spreading and so on).

\subsection{General spreading mechanism and SIS model set-up} \label{subsection:GM}
We first consider the directed simplicial complex structure and explain the physical meaning of this network structure. The directed hypergraph structure allows us to study more general dynamics compared with the undirected one. We adopt the following setting based on the physical interpretation of a spreading process. {Firstly, we recall that there is only one tail agent and the other agents are all head agents in each hyperedge.} The hyperedge describes all head agents' joint influence on the tail agent as all agents gather in the group. For example, a weight of a hyperedge of a $2$-simplex is denoted as $A_{ijk}$ and is interpreted as the joint influence of $j$ and $k$ on $i$. Here, we always set the first index in the subscript as the tail agent and all other are heads. We also assume that if the agents appear in a hyperedge, then the hyperedge with other combinations of these agents must also exist. For example, if $A_{ijk}$ exists, then $A_{jik}$, $A_{kij}$, $A_{ikj}$, $A_{jki}$, $A_{kji}$ must also exist. Since usually $A_{ijk}$ and $A_{ikj}$ represent the same influence, we set them equal. Note also that the physical meaning of $A_{jik}$ and $A_{ikj}$ are usually different, they will take usually different values. Figure \ref{fig:hyperedge} shows examples of fourth-order hyperedges.


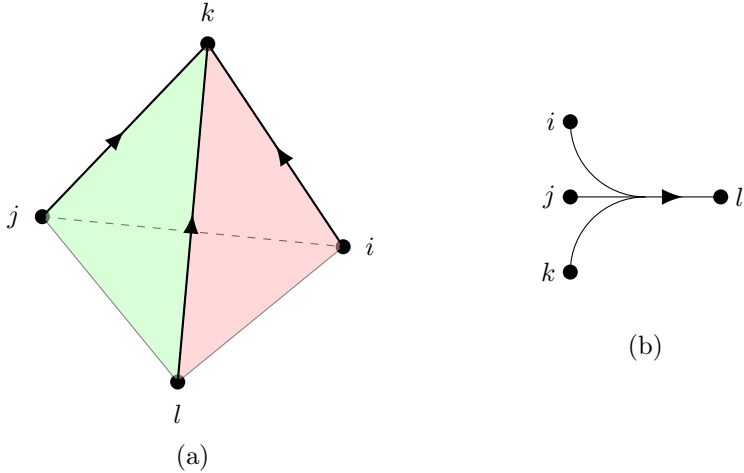
\begin{figure}
    \centering
    \begin{tikzpicture}
    \node at (-2,0){
    \begin{tikzpicture}
    \coordinate (i) at (2, -.2);
    \coordinate (j) at (-2, 0.2);
    \coordinate (k) at (0.2, 2.5);
    \coordinate (l) at (-.2, -2);
    \foreach \m in {i, j, k, l}{
        \fill[black] (\m) circle (0.1);
    }
    \draw[dashed] (j)--(i);
    \draw[-, fill = red!30, opacity = .5] (i) -- (l) -- (k) --cycle;
    \draw[-, fill = green!30, opacity = .5] (j) -- (l) -- (k) --cycle;

    \foreach \m in {i, j, k, l}{
        \fill[black] (\m) circle (0.05);
    }

    \begin{scope}[thick,decoration={
    markings,
    mark=at position 0.5 with {\pgftransformscale{1.5}\arrow{latex}}}
    ] 
    \draw[postaction={decorate}] (i)--(k);
    \draw[postaction={decorate}] (j)--(k);
    \draw[postaction={decorate}] (l)--(k);

    \node at (i) [label={[label distance=1pt]0:$i$}] {};
    \node at (j) [label={[label distance=1pt]180:$j$}] {};
    \node at (k) [label={[label distance=1pt]90:$k$}] {};
    \node at (l) [label={[label distance=1pt]270:$l$}] {};
\end{scope}
\node at (0,-3) {(a)};
    \end{tikzpicture}
    };
    \node at (4,0){
    \begin{tikzpicture}
        \fill (0,1) circle (0.1) node[left=2pt]{$i$}; 
        \fill (0,0) circle (0.1) node[left=2pt]{$j$}; 
        \fill (0,-1) circle (0.1) node[left=2pt]{$k$}; 
        \fill (2,0) circle (0.1) node[right=2pt]{$l$}; 

        \draw[-] (0,1) to[out=-90, in=180](1,0);
        \draw[-] (0,0) -- (1,0);
        \draw[-] (0,-1) to[out=90, in=180](1,0);
        \draw[postaction={decorate,decoration={
    markings,
    mark=at position 0.5 with {\pgftransformscale{2}\arrow{latex}}}}] (1,0)--(2,0);
    \node at (1,-2) {(b)};
    \end{tikzpicture}
    };
    \end{tikzpicture}
    \caption{Graphical representations of two fourth-order hyperedges. (a) A directed hyperedge $A_{kijl}$ in a simplicial complex. If $A_{kijl}$ is in the simplicial complex, then the (directed) $3$-hyperedges $A_{kjl}$, $A_{klj}$, $A_{kil}$, $A_{kli}$, $A_{kij}$ and $A_{kji}$ are also in the simplicial complex. Moreover, $A_{kil}$ further indicates that the directed $2$-hyperedges $A_{ki}$ and $A_{kl}$ are in the simplicial complex (which are, of course, conventional directed edges). The directions of $A_{jl}$, $A_{li}$, and $A_{ij}$ are not relevant for this example. (b) A (general) directed hyperedge $A_{lijk}$. We remind the reader that we only consider hyperedges with one tail, hence in this paper the first index always denotes the tail.}
    \label{fig:hyperedge}
\end{figure}

Next, in order to specify our model, we need to know about which mechanism the spreading process obeys. For example, a pairwise spreading mechanism is clear from e.g. \cite{cui2022discrete,pare2018analysis,van2008virus}. An agent can only get infected  {only if} its neighbor in the graph is infected. As for higher-order interactions, in \cite{iacopini2019simplicial}, it is the rule that if all the other (head) agents in the simplex are infected, then the (tail) agent also has a chance to get infected. However, inspired by some threshold models similar to \cite{de2020social}, we consider a more general rule of the spreading mechanism. That is, for an $n$-simplex (a hyperedge of $n$ agents), the tail agent may get infected if and only if at least $m$ of the $n-1$ head agents are infected, where $m\leq n-1$. We say the rule described above is an $m$-th-order rule of an $n$-body interaction. If $m=n-1$, we say that the interaction is a full-order $n$-body interaction. 

Now, after knowing the interaction rule, we can specify the interaction function of the spreading process. Firstly, we consider the simple case of pairwise interactions. If there is a directed edge with the weight $A_{ij}$, then the influence of $j$ on $i$ is described as an interaction function $\beta_i A_{ij} X_j(t)$, where $\beta_i$ is the intrinsic parameter of agent $i$ to describe how likely $i$ is infected when the spreading rule is satisfied; the weight $A_{ij}$ describes how close the contact between $i$ and $j$ is; $X_j(t)$ is a binary variable denoting whether $j$ is infected ($X_j=1$) or not ($X_j=0$). One can observe that this interaction is well-defined because it is non-negative and upper bounded by $1$, which means that the contribution of $j$ to the possibility of $i$ getting infected is in the interval $[0,1]$. Furthermore, as long as $j$ is healthy, $j$ will never infect $i$.

When it comes to an $m$-th order $n$-body interaction, we define the interaction function as
\begin{equation}\label{eq:interaction}
    \beta_{in} A_{ijk\dots} \frac{\Sigma^m_n}{C_{n-1}^{m}},
\end{equation}
where the subscript $n$ in $\beta_{in}$ denotes that the parameter is w.r.t $n$-body interaction, $\Sigma^m_n$ is the sum of all the polynomial combinations (ignoring the order) of $m$ binary variables of the head agents and $C_{n-1}^{m}$ is the combination $m$ out of $n-1$. For example, if we consider a second-order $4$-body interaction, we get the interaction function of 
\begin{equation}\label{eq:aijkl}
    \beta_{i4} A_{ijkl} \frac{X_{j}X_{k}+X_{j}X_{l}+X_{k}X_{l}}{C_{3}^{2}}.
\end{equation}

In the case of the full-order $(m+1)$-body interaction, the interaction function is  $\beta_{i(m+1)} A_{ijk\dots} X_j X_k \cdots$. Thus, one can use multiple lower-order hyperedges where the interaction is full-order to replace the higher-order hyperedge where the interaction is not full-order. 
\begin{lemma}\label{lem:eq}
    Any not-full-order ($l$-order, $l<m$) ($m+1$)-body interaction \eqref{eq:interaction} can be replaced by a composition of multiple full-order $(l+1)$-body interactions, where $l<m$.
\end{lemma}
\begin{pf}
    Firstly, consider the equation \eqref{eq:interaction}. If the interaction is not in full order, $\Sigma^l_n$ is a composition of several products among $l$ variables. Each of these products corresponds to a full-order $(l+1)$-body interaction. By rearranging the weights, any not-full-order ($m+1$)-body interactions \eqref{eq:interaction} can be replaced by a composition of multiple full-order $(l+1)$-body interactions, where $l<m$.
\end{pf}
For instance, $A_{ijkl}$ in the interaction function \eqref{eq:aijkl} can be replaced by a combination of $A_{ijk}$, $A_{ikj}$, $A_{ikl}$, $A_{ilk}$, $A_{ijl}$, $A_{ilj}$. Due to the simplicial complex structure, if $A_{ijkl}$ exists, all the links $A_{ijk}$, $A_{ikj}$, $A_{ikl}$, $A_{ilk}$, $A_{ijl}$, $A_{ilj}$ exist too. By choosing again a proper weight, one only needs to consider full-order interactions under this choice of the interaction function. Note that, by doing so, the model originally based on a simplicial complex is now constructed on a hypergraph instead, because some lower-order interaction may be absent while the corresponding higher-order interaction may still exist after the replacement.

The choice of interaction function is not unique. In order to make sure it is well-defined, one only needs to ensure that
\begin{itemize}
    \item it is zero when the corresponding rule is not satisfied;
    \item otherwise, it is positive but upper bounded by $1$.
\end{itemize}
In this paper, we firstly focus on a concrete system where we use the polynomial form $\beta_{in} A_{ijk\dots} \frac{\Sigma}{C_{n}^{m}}$ as interaction functions. Afterwards, we further study a more general case when the interaction function takes an abstract form. If other types of interactions are considered, the equivalence between one higher-order interaction and multiple lower-order interactions may not be valid anymore, depending on the function.

By adopting the interaction function \eqref{eq:interaction} and inspired by \cite{iacopini2019simplicial}, we model an epidemic process as a continuous-time Markov chain. For simplicity, we only consider the full-order interaction by adopting the interaction function introduced above; however, bear in mind that it indeed contains all other lower-order interactions and is just replaced by the full-order interaction on lower-order hyperedges. Although our model is firstly developed on a simplicial complex structure, after the replacement, the model is built on a hypergraph instead. If every hyperedge originally incorporates only the full-order interaction, then the simplicial complex structure remains unchanged.

Our model is derived by considering all these interactions \eqref{eq:interaction} with different orders. The state transitions in the SIS model can be characterized by a Markov chain
with $2^n$ states, where $n$ is the number of nodes in the simplicial complex. Specifically, we denote the
random variable $X_i(t)$ as the state of node $i$ at time $t$. Moreover, $X_i = 0$ if node $i$ is in state $S$ (suceptible) and $X_i = 1$ if node $i$ is in state $I$ (infected). The infection process is described as the interaction above; whereas, the curing process is assumed to be passive with the rate $\delta_i$.
The local nodal dynamics can be modeled as a two-state Markov chain in continuous time. For a sufficiently small time $\Delta t$, it holds for the infection process that

\begin{equation}\label{eq::re1}
\begin{split}
    &\textbf{Prob}[X_i(t+\Delta t)=1|X_i(t)=0]=\sum_{j\in \mathbf{N}_2^{\text{i}}} \beta_i A_{ij} X_j(t)\Delta t\\ 
    &+\sum_{j,k\in \mathbf{N}_3^{\text{i}}} \beta_{i3} A_{ijk} X_j(t)X_k(t)\Delta t
    +\cdots+\sum_{j,k,l,\cdots \in \mathbf{N}_n^{\text{i}}} \beta_{in} A_{ijkl\cdots} X_j(t)X_k(t)X_l(t)\cdots \Delta t
\end{split}
\end{equation}
where the notation $\mathbf{N}_2^{\text{i}}$ denotes the set including all the conventional edges where the node $i$ is the tail, the notation $\mathbf{N}_n^{\text{i}}$ describes the set including all the $2$-simplexes where the node $i$ is the tail, and the $n$-body interaction within the $n$-simplex is of the full order.

For the curing process, it holds that
\begin{equation} \label{eq::deltanodal}
\textbf{Prob}[X_i(t+\Delta t)=0|X_i(t)=1]=\delta_i X_i(t) \Delta t.
\end{equation}

We now apply the mean-field approximation, leading to
\[\sum_{j\in \mathbf{N}_2^{\text{i}}} \beta_i A_{ij} X_j(t)= \sum_{j\in \mathbf{N}_2^{\text{i}}} \beta_i A_{ij} \mathbf{E}(X_j(t)),\] 
\[\sum_{j,k\in \mathbf{N}_3^{\text{i2}}} \beta_{i3} A_{ijk} X_j(t)X_k(t)= \sum_{j,k\in \mathbf{N}_3^{\text{i2}}} \beta_{i3} A_{ijk} \mathbf{E}(X_j(t)X_k(t)).\] This approximation is valid and accurate if the states of the neighbors are sufficiently independent and the number of the in-neighbors of node $i$ is large.

Thus, the central limit theorem \cite{lukacs2014probability} applies and it follows that when $\Delta t \rightarrow 0$, by adopting \eqref{eq::re1}, one gets
\begin{equation}
    \begin{split}
        & \frac{d\mathbf{E}(X_i(t))}{dt}
        = \mathbf{E} \Big( \Big[\sum_{j\in \mathbf{N}_2^{\text{i}}} \beta_i A_{ij} X_j(t) 
        +\sum_{j,k\in \mathbf{N}_3^{\text{i}}} \beta_{i3} A_{ijk} X_j(t)X_k(t) \\
        &+\sum_{j,k,l,\cdots\in \mathbf{N}_n^{\text{i}}} \beta_{in} A_{ijkl\cdots} X_j(t)X_k(t)X_l(t) \Big](1-X_i(t))-\delta_i X_i(t)
        \Big).
    \end{split}
\end{equation}

Notice that $\mathbf{E}(X_i(t)X_j(t))=\mathbf{E}(X_i(t))\mathbf{E}(X_j(t))-\mathbf{cov}(X_i(t),X_j(t))$. For the large-scale network, we adopt the independence assumption \cite{8954786,pare2018virus} so that \\ $\mathbf{cov}(X_i(t),X_j(t)))=0$. {The article \cite{van2015accuracy} indicates that the accuracy of the approximation increases with the network scale by using an SIS model on a conventional graph. It shows that this approximation is generally sufficiently accurate on large-scale networks and the error is usually tolerable on a network with several hundred nodes. In \cite{iacopini2019simplicial}, they show that a mean-field model captures quantitative features of the origin stochastic process using a simplicial complex setting.}
By ignoring all covariance and letting $x_i(t)= \textbf{Prob}[X_i(t)=1]$, we get
\begin{equation}\label{eq::syssimp1}
\begin{split}
\dot{x}_i=
&(1-x_i)\left(\sum_{j\in \mathbf{N}_2^{\text{i}}} \beta_i A_{ij} x_j\right) 
+(1-x_i)\left(\sum_{j,k\in \mathbf{N}_3^{\text{i}}} \beta_{i3} A_{ijk} x_jx_k\right)+\cdots\\
&+(1-x_i)\left(\sum_{j,k,l,\cdots\in \mathbf{N}_n^{\text{i}}} \beta_{in} A_{ijkl\cdots} x_jx_kx_l \cdots\right)-\delta_i x_i,
\end{split}
\end{equation}
which is similar to \textbf{the general higher-order SIS model} in \cite{cisneros2021multi} but now the parameter of $\beta_i,\beta_{i3},\cdots, \beta_{in}$ can be heterogeneous, and our model is actually based on a directed hypergraph. In contrast to \cite{cisneros2021multi}, we now provide an explanation, at the microscopic level, of the physical meaning of our model, the mechanism behind the modeling framework, and how it relates to the stochastic epidemic model. Furthermore, we explore the connection between the epidemic model on a general hypergraph and on a simplicial complex. The spreading process on a simplicial complex may become a process on a hypergraph when one assumes that the different spreading mechanisms occur on different hyperedges.

In \cite{iacopini2019simplicial}, the authors use a simplicial structure but with undirected and equally-weighted hyperedges to construct the stochastic process. They only consider a full-order interaction. If one assumes that the weights of the hypergraph are the same in a group of the same agents and chooses the same interaction function above \eqref{eq:interaction}, then our model is equivalent to the SCM (Social Contagion Model) model proposed in \cite{iacopini2019simplicial}. It is worthwhile mentioning that our model does consider any-order interaction, but because of the equivalence result we establish, they are implicit in the lower-order hyperedge. Because diverse kinds of interactions are considered, our model is eventually equivalent to a model on a general hypergraph according to the Lemma \ref{lem:eq}. {In \cite{iacopini2019simplicial}, it shows that the mean-field model is indeed quantitatively comparable to the original stochastic process using simplicial complexes with several hundred nodes.}

Since achieving analytical results for the arbitrary higher-order case is challenging and tedious, we first consider the hypergraph structure up to hyperedges of degree 2 (three nodes involved in a hyperedge). In such a case, the dynamics of \eqref{eq::syssimp1} read as
\begin{equation}\label{eq::syssimp2}
\begin{split}
\dot{x}_i=
&(1-x_i)\left(\sum_{j\in \mathbf{N}_2^{\text{i}}} \beta_i A_{ij} x_j\right) 
+(1-x_i)\left(\sum_{j,k\in \mathbf{N}_3^{\text{i}}} \beta_{i3} A_{ijk} x_jx_k\right)
-\delta_i x_i,
\end{split}
\end{equation}
which is similar to the \textbf{simplicial SIS model} in \cite{cisneros2021multi} but here the parameters of $\beta_i$ and $\beta_{i3}$ can be heterogeneous. Just to recall, we now only consider the hyperedge up to 3 bodies. However, the results based on a hypergraph with even higher-order hyperedges will be still provided later when we consider a general case in section \ref{sec:abstract}.


Finally, although the framework so far describes a direct spreading process reasonably well, it cannot characterize indirect spreading. For example, one can get infected not only because one gets in contact with infected neighbors but also because one uses a resource that carries the pathogen of the disease. A common way to overcome this drawback is to introduce coupled pathogen dynamics with a resource. In the next subsection, we will introduce how to couple the system that we have proposed so far with a pathogen evolution.


\subsection{Higher-order SIWS Model} \label{sec::SIWS}

Consider an SIS-type dynamics over a human contact hypergraph consisting of $n$ groups of individuals 
(agents) and the pathogen diffuses over $m$ different types of resources that are utilized by some agents in the human contact network. An individual can be infected by the disease either through its infected neighbor in the human contact network or the polluted resource the individual utilizes. In turn, the resource can be contaminated by the infected agents using the resource.

By adding a $W$ (Water) compartment, we propose a higher-order SIWS model. We emphasize that the compartment $W$  can refer to any other media. 

In this setting, the epidemic process and pathogen dynamics in the resource evolve as
\begin{align}
&\dot{x}_i= -\delta_i x_i 
+ (1-x_i) \left(\sum_{j=1}^{n} \beta_{ij} x_j
 +\sum_{j=1}^{m} \beta_{ij}^{w} w_j \right)  \label{eq::sys_sis_dt_x}\\
 &\qquad+ (1-x_i) \left(\sum_{j,k\in \mathbf{N}_3^{\text{i}}} \beta_{ijk} x_ix_j
 +\sum_{j,k\in \mathbf{N}_3^{\text{i}}} \beta_{ijk}^{wx} w_jx_k 
 +\sum_{j,k\in \mathbf{N}_3^{\text{i}}} \beta_{ijk}^{xw} x_jw_k
 \right), \notag\\
 &\dot{w}_j= -\delta_{j}^{w}w_{j}+\sum_{k=1}^{n}c_{jk}^{w}x_{k}+\sum_{i,k\in \mathbf{N}_3^{\text{i}}}c_{jik}^{xx}x_{i}x_{k},\label{eq::sys_sis_dt_w}
\end{align}
where $i=1,\ldots,n$, $j=1,\ldots,m$, and all notations and parameters are defined in table \ref{tab:notation}. We remind the reader that for now we consider interactions of up to 3 bodies.
\begin{table}[htbp]
\caption{Notations and parameters in equations \eqref{eq::sys_sis_dt_x}-\eqref{eq::sys_sis_dt_w}}
\label{tab:notation}
{\color{black}\begin{tabular}{@{}ll@{}}
\hline
$x_i=x_i(t)$ & Infection probability of agent $i$ \\ & or infection proportion of group $i$ at time $t$  \\ 
$w_{j}=w_j(t)$ & Pathogen concentration in the $j$-th resource at time $t$\\
$\delta_{i}$ & Healing rate of node $i$\\
$\delta_{j}^{w}$ & Decaying rate of the pathogen in resource $j$\\
$A_{ij}$ & Weight of the conventional edges (first-order hyperedges) \\ & in the corresponding $n$-node simplicial $\mathcal{G}(A)$  \\
$A_{ijk}$ & Weight of the second-order hyperedges \\ & in the corresponding $n$-node simplicial $\mathcal{G}(A)$\\
$\beta_{i}$ & First-order infection rate of node $i$\\
$\beta^{xx}_{i}$ & Second-order infection rate of node $i$\\
$\beta_{ij} := \beta_{i} A_{ij}$ & First-order effective transmission rate from node $j$ to node $i$  \\
$\beta_{ijk} := \beta_{i}^{xx} A_{ijk}$ &   Second-order effective transmission rate \\ & from node $j$ and $k$ to node $i$ \\
$c_{jk}^{w}$ & First-order person-resource contact rate of agent $j$ \\ &  with resource $k$ \\
$c_{ijk}^{xx}$ & Second-order person-resource contact rate \\ & of agent $i$ and $k$ with  resource $j$ \\
$\beta_{ij}^{w}$ & First-order resource-person transmission rate of agent $i$ \\ &  with resource $j$ \\
$\beta_{ijk}^{wx},\beta_{ijk}^{xw}$ & Second-order human-infrastructure co-infection rate \\ & of node $j$ and $k$ to node $i$, the superscript $xw$ ($wx$) indicates $j$ ($k$) \\ &is a human index and $k$ ($j$) is a resource index \\
\hline
\end{tabular}}
\end{table}

For \eqref{eq::sys_sis_dt_x}-\eqref{eq::sys_sis_dt_w}, we make the assumption that no direct interaction exists between resources because we mainly focus on the pathogen (information) diffusing over different types of media. The inner interaction inside the same type of medium can be ignored when we model the same type of medium as a whole. Since we also assume that the different types of media are sufficiently independent, we do not need to consider the higher-order interaction purely induced by resources. In the case of the higher-order interaction where both resource and agent are present, we assume that the higher-order interaction takes a similar form as when only agents are present, i.e., there are no $xw$-cross-terms, and thus the resource remains independent of each other.

Due to the fact that the indirect spreading in the resource is a macroscopic process and the concentration cannot be simply modeled as a binary variable, it is not suitable to model the concentration in the Markov chain. However, the idea of coupling the epidemic dynamics in the human contact network with the pathogen dynamics is widely adopted and supported in e.g. \cite{pare2022multi,janson2020networked,9029305,tien2010multiple}.

It is worthwhile mentioning that the SIWS model usually describes a process of water-borne disease or any other virus that can spread over a medium other than the human-to-human pathway. So in the above, we also explain the Simplicial SIWS model in this context. However, notice that a real virus spreading usually obeys the rule that one may get infected if and only if one of the agents in the (hyper)-edge is infected, and thus all the interactions are of the first order (pairwise). That is to say, the conventional graph structure is sufficient to study such a case if we assume the interaction takes a polynomial form. Interested readers may refer to the literature about SIWS models on the conventional graphs in discrete time \cite{cui2022discrete} or in continuous time \cite{pare2022multi,janson2020networked,9029305}. 

In the context of \emph{information diffusion}, things are different. Firstly, the higher-order interaction really counts. A real spreading process is a hybrid process of interactions with all orders. 
{One may obey different rules under different situations. For instance, one will tend to believe the information if half of the neighbors in the hyperedge believe the information, since all the neighbours are his close friends. However, if he is less familiar with the neighbors, he may tend to believe the information only unless all of the neighbors in the hyperedge believe the information.}
Secondly, the resource compartment also plays a role in information diffusion. While the direct information spreading usually corresponds to private chat (pairwise) or group chat (higher-order) via Whatsapp or some similar apps, the indirect spreading are usually the posts on Twitter or videos on Tiktok. The posts or videos also influence the agents to an extent depending on how popular the posts or videos are. Thus, such posts or videos are resources, like water, and their popularity and influence behave like the pathogen concentration correlated to the infection level. Once the infection level is high, the post is getting popular and in turn promote the spreading of the information. The higher-order interaction is also necessary when the resource gets involved here. For example, when one sees an interesting video on Tiktok, he or she may directly forward it to a friend and have a discussion. Both the content of the video and his or her conversation with a friend will jointly influence how strongly one believes in the information.

To summarize, in this section, we have proposed a general model describing a spreading process on a hypergraph up to 3-body interactions. Our model considers any higher-order interaction and a general rule for spreading. The model we have proposed can be used, for example, in the context of the physical virus spreading and information diffusion. Furthermore, we take into account both, the direct and indirect spreading processes. The rest of the paper aims at achieving analytical results for the higher-order SIWS model and its potential extension to deepen our understanding of a general spreading process.

\section{Analytical Results} \label{sec::main}

Let us now consider a higher-order SIWS model on a hypergraph of $n$ populations and $m$ resources.

Component-wise, the system we consider is given by \eqref{eq::sys_sis_dt_x} and \eqref{eq::sys_sis_dt_w}. We can further derive it in a matrix form as
\begin{equation}\label{eq::sys_siws_hyp_z}
    \dot{z}=-D_f z+(I-Z) (B_f z+ H(z));
\end{equation}
where 
\begin{equation}\label{eq::notations1}
  \begin{aligned}
    &    D_{f}:=\left[ 
        \begin{matrix}
	D & \mathbf{0}  \\
	\mathbf{0} & D^{w}  \\
	\end{matrix}
        \right],\;
     Z=\left[\begin{matrix}
	\Dg(x) & 0 \\
0& 0  \\
	\end{matrix}\right],\;
	B_f=\left[\begin{matrix}
	B & B_w \\
C_w& 0  \\
	\end{matrix}\right],\\
&	H(z)= [z^\top B_{f1} z,\cdots, z^\top B_{fn+m} z]^\top;
  \end{aligned}
\end{equation}
and 
\begin{equation}
    \begin{split}
        D&=\Dg((\delta_1,\cdots,\delta_n)^\top)\\
        D_w&=\Dg((\delta^w_1,\cdots,\delta^w_m)^\top)\\
        B_w&=[\beta^w_{ij}]_{n \times m}\\
        B&=[\beta_{ij}]_{n \times n}\\
        C_{w}&=[c^{w}_{jk}]_{m \times n}
    \end{split}
\end{equation}

Furthermore, if $i\leq n$, the matrix $B_{fi}$ is defined by 
\begin{equation}
\begin{split}
B_{fi}&=\left[\begin{matrix}
	[\beta_{ijk}]_{n \times n} & [\beta^{xw}_{ijk}]_{n \times m} \\
[\beta^{wx}_{ijk}]_{m \times n}& \mathbf{0}  \\
	\end{matrix}\right],  \qquad \text{ if } i\leq n\\
 B_{fi}&=\left[\begin{matrix}
	[c^{xx}_{ijk}]_{n \times n} & \mathbf{0} \\
\mathbf{0}& \mathbf{0}  \\
	\end{matrix}\right], \qquad \text{ if } i\geq n+1.
\end{split}
\end{equation}

Now, we present the analytical results of the general SIWS on a hypergraph \eqref{eq::sys_siws_hyp_z}. Firstly, we need the following assumptions to ensure that our model is well-defined.

\begin{assumption}\label{ass:xini1}
The initial condition of the infection level satisfies $x_i(0)\in [0,1]$ for all $i=1,\cdots,n$.
\end{assumption}

\begin{remark}
Since $x_i(0)\in [0,1]$, according to \eqref{eq::sys_sis_dt_w}, 
\begin{equation} \label{eq:wjmax}
    w_{j}\leq w_{j\max}\coloneqq\frac{\sum_{k=1}^{n}c_{jk}^{w}+\sum_{i,k\in \mathbf{N}_3^{\text{i}}}c_{jik}^{xx}}{\delta_j^w}.
\end{equation}
\end{remark}

\begin{assumption}\label{ass:para} 
The matrices $B_f$, $B_{fi}$ are non-negative, $i=1,\cdots,n+m$. The diagonal entries of $D_f$ are positive.
\end{assumption}

\begin{remark}
    We define later the reproduction number of the higher-order spreading as $\rho(D_f^{-1} B_f)$. To ensure the reproduction number is well-defined, $D_f$ must be positive.
\end{remark}


In the following Lemma, we describe the general properties of \eqref{eq::sys_siws_hyp_z}. 
\begin{remark}
For simplicity, we define the notation $F(z)=B_f z$, $f_i=F(z)_i=(B_f z)_i$, $h_i=H(z)_i=z^\top B_{fi} z$. Then, the concrete system \eqref{eq::sys_siws_hyp_z} is a special case of the general system \eqref{eq::sys_siws_hyp_z2}, which is discussed later. Moreover, the concrete interaction function defined by $F(z)=B_f z,\, H(z)_i= z^\top B_{fi} z$ satisfies Assumptions \ref{ass:f0}-\ref{ass:irre3} for the general system \eqref{eq::sys_siws_hyp_z2}.
\end{remark}
\begin{lemma}[General properties]\label{lem:GP}
Consider the system \eqref{eq::sys_siws_hyp_z}. Let $w_{j\max}$ be as in \eqref{eq:wjmax}, $\mathbf{D}=\left\{z=(x,w)^\top=(x_1,\ldots,x_n,w_1,\ldots,w_n)^\top\,|\, x_i\in(0,1),\,w_j\in(0,w_{j\max}) \right\}$
and $\mathbf{\Bar{D}}$ be its closure. If Assumptions \ref{ass:xini1}-\ref{ass:para} hold, then
\begin{itemize}
    \item [i)] the set $\mathbf{\Bar{D}}$ is positively invariant;
    \item [ii)] the set $\mathbf{D}$ is also positively invariant;
    \item [iii)] the origin is always an equilibrium of the model and there is no other equilibrium on the boundary of $\mathbf{D}$.
\end{itemize}
\end{lemma}


\begin{proof}
This lemma is a special case of Lemma \ref{lem:gpa}. Thus, the proof is omitted here and we provide the proof of a more general case in Lemma \ref{lem:gpa}.
\end{proof}

\begin{remark}
The general properties of Lemma \ref{lem:GP} show that the system is well-defined. The set $\mathbf{\Bar{D}}$ is the system's domain for \eqref{eq::sys_siws_hyp_z}.
\end{remark}

Now, we need one more assumption to achieve analytical results.
\begin{assumption}\label{ass:irre}
The matrix $B_f$ is irreducible.
\end{assumption}

\begin{remark}
Assumption \ref{ass:irre} implies that the graph of pairwise interactions is \\
strongly connected. Our model is originally based on a simplicial complex structure. If any two nodes in the simplicial complex can be connected with a sequence of edges (or hyperedges), then the graph of pairwise interaction is strongly connected. Since we have performed some replacements introduced in Section \ref{subsection:GM}, the structure of the simplicial complex may break. However, since the graph of pairwise interactions is at the lowest level, the property of its strong connectivity does not break after the replacements. Hence, the assumption is rather natural.
\end{remark}

It is easy to check that the system \eqref{eq::sys_siws_hyp_z} has an all-zero equilibrium, which is the so-called \emph{healthy state}. 

\begin{theorem}[Healthy-state behaviour] \label{thm:healthysig}
Consider the system \eqref{eq::sys_siws_hyp_z}. If Assumptions \ref{ass:xini1}-\ref{ass:irre} hold, then
the healthy state (zero equilibrium) always exists and is:
\begin{enumerate}
    \item locally stable if $\rho(D_f^{-1} B_f)<1$,
    \item  globally exponentially stable if  $\rho(D^{-1}_f B_f+D^{-1}_f[u^{\top}B_{f1},\cdots,u^{\top}B_{fn+m}]^{\top})<1$, where $u=[1,\cdots,1,w_{1\max},\cdots,w_{m\max}]^{\top}$
    \item unstable if $\rho(D^{-1}_f B_f)>1$.
\end{enumerate}
\end{theorem}
 \begin{proof}
 The local stability and instability of the healthy state is a special case of Theorem \ref{thm:healthy}. The proof is omitted here and we provide a proof for a more general case in Theorem \ref{thm:healthy}.
 
 Next, we prove global stability. It is straightforward to check that 
 \begin{equation}\label{eq:matrix}
     D^{-1}_f B_f+D^{-1}_f[u^{\top}B_{f1},\cdots,u^{\top}B_{fn+m}]^{\top}
 \end{equation}
  is an irreducible non-negative matrix. Let $v$ be the strictly positive left Perron-Frobenius eigenvector of the matrix \eqref{eq:matrix} and $\lambda$ be the corresponding Perron-Frobenius eigenvalue. Define the Lyuapunov function $V=v^{\top}D_f^{-1}z$. We observe that $V\geq 0$ and $V=0$ if and only if $z=\mathbf{0}$. Let $\delta=(\delta_1, \cdots, \delta_n)^\top$ and  $\delta^w=(\delta^w_1, \cdots, \delta^w_m)^\top$. Then, we get
 {\begin{equation}\label{eq:v}
     \begin{split}
         \dot{V}&=-v^{\top}D_f^{-1}(D_f z+(I-Z)B_f z\\
         &+(I-Z)[z^{\top}B_{f1}z,\cdots,z^{\top}B_{fn+m}z]^{\top})\\
         &\leq -v^{\top} z + v^{\top}(D_f^{-1}B_f z\\
         &+ D^{-1}_f[u^{\top}B_{f1} z,\cdots,u^{\top}B_{fn+m} z)^{\top}]\\
         & \leq (-1+\lambda) v^{\top} D_f D^{-1}_f z
         \leq (-1+\lambda) \min(\delta, \delta^w)V
     \end{split}
 \end{equation}}
 In view of the comparison lemma \cite{khalil2002nonlinear}, we have $0\leq V(t) \leq V(0)e^{(-1+\lambda) \min(\delta, \delta^w) t}$, which further implies that
 $x_i\leq \frac{\max(\delta^{-1}, (\delta^{w})^{-1})v^{\top}z(0)}{\min(\delta^{-1}, (\delta^{w})^{-1})v_i} e^{(-1+\lambda) \min(\delta, \delta^w) t} $
 and \\
 $w_j\leq \frac{\max(\delta^{-1}, (\delta^{w})^{-1})v^{\top}z(0)}{\min(\delta^{-1}, (\delta^{w})^{-1})v_{j+n}} e^{(-1+\lambda) \min(\delta, \delta^w) t}.$ Thus, we complete the proof.

 \end{proof}

 \begin{remark}
 In epidemiology, the reproduction number (usually denoted as $R_0$) is defined to be the expected number of cases directly generated by one infected case in a population. Generally, if $R_0<1$, then the virus gradually dies out, while if $R_0>1$, the pandemic breaks out. As for our model \eqref{eq::sys_siws_hyp_z}, we suggest directly using $R_0=\rho(D_f^{-1} B_f)$ as the reproduction number, since it indicates whether the healthy state is locally stable or not. Naturally, the behavior of the higher-order system \eqref{eq::sys_siws_hyp_z} is different from the SIWS model on a conventional graph studied by \cite{pare2022multi,janson2020networked,9029305}. In the following, we show that for the system \eqref{eq::sys_siws_hyp_z}, even if $R_0<1$, the pandemic may still break out.
 \end{remark}
 
Other than the zero equilibrium, the system also processes a strictly positive equilibrium, which is called \emph{endemic equilibrium}.

\begin{theorem}(Existence of an endemic equilibrium with high infection levels)\label{thm:exend}
Consider the system \eqref{eq::sys_siws_hyp_z}. If Assumptions \ref{ass:xini1}-\ref{ass:irre} hold, then there exists an endemic equilibrium $z^*\gg \mathbf{0}$ such that $x_i^*\geq \frac{2}{3}$ and $w_i^*\geq \frac{2w_{i\max}}{3}$ for any $i$ such that $B_{fi}\neq \mathbf{0}$ 
whenever 
$\theta=\min_{i} \{(D^{-1}_f B_f)_i+[D^{-1}_f(\tilde{u}^{\top}B_{f1},\cdots,\tilde{u}^{\top}B_{fn+m})^\top]_i\}\geq 4.5$, where $\Tilde{u}=[\Tilde{u}_1,\Tilde{u}_2,\cdots,\Tilde{u}_{n+m}]^\top$ is given by
\begin{equation*}
       \tilde{u}_i = \begin{cases}
                        1,  & B_{fi} \neq \mathbf{0},\qquad i=1,\ldots,n \\
                        w_{i\max}, & B_{fi}\neq\mathbf{0}, \qquad i=n+1,\ldots,n+m\\
                        0,  & B_{fi}=\mathbf{0}
    \end{cases}
\end{equation*}
\end{theorem} 

\begin{proof}
Define the map 
 \begin{equation} \label{eq:map}
 \begin{split}
    \hat{T}(z)&=T(D^{-1}_f F(z)+D^{-1}_f H(z));\\
 \end{split}
 \end{equation}
where $T_i=\frac{p_i}{1+p_i}$for $i=1,\ldots,n$ and $T_j=p_i$ for $i=n+1,\ldots,n+m$ with $p_i=(D^{-1}_f F(z)+D^{-1}_f H(z))_i$.
 
Next, we calculate the endemic equilibrium $z^*=(x^*,w^*)^\top=(x_1^*,\ldots,x_n^*,w_1^*\ldots,w_m^*)^\top$ by setting $\dot{x}_i$ and $\dot{w}_j$ to zero. 
This leads to $x_{i}^*=\frac{p_i(z^*)}{1+p_i(z^*)}$ and $w^*_j=p_{j+n}(z^*)$,
which implies that the endemic equilibrium $z^*$ is a fixed point of $\hat{T}$, i.e., $\hat{T}(z^*)=z^*$.
We know that both $\frac{p}{1+p}$ and $p$ are monotonically increasing in $z\geq 0$. Thus, the maps $\hat{T}(p(z))$ and $T(p(z))$ are monotonically increasing in $z\geq 0$. Moreover, if $\alpha>1$, $p\leq 1-\frac{1}{\alpha}$ with $p\geq 0$ if and only if $\alpha\geq \frac{1}{1-p}$. If $\alpha>1$ and $p>0$, $\alpha p > p$. Thus, we have $T(\alpha p(z))\geq p(z)$ with $p(z)\geq \mathbf{0}$ if and only if $p_i\leq 1-\frac{1}{\alpha}$ for $i\leq n$. 

Then, we have that
\begin{equation*}
\begin{split}
    \hat{T}(z)&=T(D^{-1}_f B_f z+D^{-1}_f[z^{\top}B_{f1}z,z^{\top}B_{f2}z,\cdots]^{\top})\\
    &\geq T\left(\frac{2}{3}D^{-1}_f B_f \tilde u+\frac{4}{9}D_f^{-1}[\tilde u^{\top}B_{f1}\tilde u,\tilde u^{\top}B_{f2}\tilde u,\cdots]^{\top}\right)\\
    &\geq T\left(\frac{4}{9}\theta \tilde u\right)\geq \frac{2}{3} \tilde u.
\end{split}
\end{equation*}
By using the monotonicity, we derive that $u$ must be the upper bound. In other words, $\hat{T}$ maps the set 
\begin{equation}
    \begin{split}
        \Big\{z&=(x,w)^\top =(x_1,\cdots,x_n,w_1,\cdots,w_m)^\top\,|\\
        &   \frac{2}{3}\tilde u_i \leq x_i \leq 1 , \frac{2w_{j\max}}{3}\tilde u_{j+n}\leq w_j \leq w_{j\max},\\
        & i=1,\cdots,n,\, j=1,\cdots, m \Big\}
    \end{split}
\end{equation}
into itself. By using Brouwer's Fixed-Point Theorem \cite{shapiro2016fixed}, we complete the proof. 
\end{proof}

Furthermore, we can show that the endemic equilibrium in Theorem \ref{thm:exend} is locally stable.

\begin{theorem}(Bi-stability may occur)
Consider the model \eqref{eq::sys_siws_hyp_z}.  If Assumptions \ref{ass:xini1}-\ref{ass:irre} hold and all $B_{fi}\neq \mathbf{0}$, any endemic equilibriunm in the form of $z^*\gg \mathbf{0}$ such that $\frac{2}{3} \leq x_i^*\leq 1$ and $\frac{2w_{\max}}{3}\leq w_j^*\leq w_{\max}$ is locally stable.
\end{theorem}

\begin{proof}
Firstly, we calculate the Jacobian for the whole system:
{\small\begin{equation*}
\begin{split}
    J&=-D_f+(I-{Z}) \Big(B_f-\hat{I}\Dg([z^\top B_{f1} z,\cdots,z^\top B_{fn+m} z]^\top\\&+B_f z
    +[z^\top(B_{f1}+B_{f1}^\top),\cdots,z^\top(B_{fn+m}+B_{fn+m}^\top)]^\top)
    \Big). 
\end{split}
\end{equation*}}
 It is easy to check that $J$ is an irreducible Metzler matrix.

Next, we define the vector $\tilde{z}^*=(\frac{x^*}{2},w^*)^\top$. Then,
{\begin{equation*}
\begin{split}
 J\tilde{z}^*&=-D_f \tilde{z}^*+(I-{Z}) (B_f \tilde{z}^*\\& +[z^\top(B_{f1}+B_{f1}^\top)\tilde{z}^*,\cdots,z^\top(B_{fn+m}+B_{fn+m}^\top)\tilde{z}^*]^\top)\\
 &-\hat{I}(\Dg([z^\top B_{f1} z,\cdots,z^\top B_{fn+m} z]^\top+B_f z)\tilde{z}^*)
\end{split}
\end{equation*}}
with $\hat{I}=\left[\begin{matrix}
	I & \mathbf{0} \\
\mathbf{0}& \mathbf{0}  \\
	\end{matrix}\right]$.

Element-wise, we get the following.
For $n+1 \leq i\leq n+m$,
\begin{equation*}
\begin{split}
    (J(z^*)\tilde{z}^*)_i&=-\delta^w_{i-n} w^*_{i-n}+\sum_{j}c^w_{(i-n)j} \frac{x^*_j}{2}\\&+2\sum_{i,j}c^{xx}_{(i-n)jk} x^*_j\frac{x^*_k}{2}\\
    &=-\sum_{j}c^w_{(i-n)j} \frac{x^*_j}{2}<-d_i w^*_{i-n}= -d_i z_i^*,
\end{split}
\end{equation*}
with some $d_i>0$. 

On the other hand, for $i\leq n$,
{\begin{equation*}
\begin{split}
    (J(z^*)\tilde{z}^*)_i&=-\frac{x_i^*}{2}(\sum_{j}\beta_{ij}x^*_j)+(\frac{1}{2}-x^*_i)\sum_{j}\beta_{ij}^w \frac{w^*_j}{2}\\
    &+\frac{1}{2}(2-3x^*_i)(\sum_{j,k}\beta^{wx}_{ijk}x^*_j w^*_k+\sum_{j,k}\beta^{wx}_{ijk}x^*_k w^*_j)\\
    &+(\frac{1}{2}-x^*_i)\sum_{j,k}\beta^{xx}_{ijk}x^*_j x^*_k.
\end{split}
\end{equation*}}
Since $2-3x_i^*<0$ and $\frac{1}{2}-x^*_i<0$, we observe that  $(J\tilde{z}^*)_i<-d_i \frac{z_i^*}{2}$ for $i\leq n$ with some $d_i>0$. By combining the results from the case $n+1 \leq i\leq n+m$, we have $J\tilde{z}^*<-d \tilde{z}^*$ with some $d>0$. According to Theorem 10.14 in \cite{FB-LNS}, we confirm that $J$ is Hurwitz and the endemic equilibrium we consider here is locally stable.
\end{proof}

Notice that the condition guaranteeing the existence of an endemic equilibrium in the form of Theorem \ref{thm:exend} and its local stability can hold together with the condition that the reproduction number is less than $1$. This is an interesting phenomenon of the higher-order SIWS model on hypergraphs because we see that the system may have a multi-stability domain when the reproduction number is smaller than $1$. In such a domain, the healthy state and the endemic equilibrium in the form of Theorem \ref{thm:exend} are locally stable. Different from the higher order system, the SIWS on a conventional graph has one unique globally stable equilibrium, healthy state, when the reproduction number is less than $1$. We will give some numerical examples of this case later in the section on numerical examples (Figure \ref{fig:bistability}).


\begin{theorem} (Existence of an endemic equilibrium when $R_0>1$)
Consider the model \eqref{eq::sys_siws_hyp_z}, if Assumptions \ref{ass:xini1}-\ref{ass:irre} hold, there exists an endemic equilibrium $z^*\gg \mathbf{0}$ if $R_0=\rho=\rho(D^{-1}_f B_f)>1$.
\end{theorem}

\begin{proof}
Define the set $\Tilde{Y}=\{z=(x,w)^{\top}|c_i\leq x_i\leq 1, c_{j+n}\leq w_j \leq w_{j\max} \}$, where $c=(c_1,\cdots,c_{n+m})^{\top}=\alpha v$, $v$ is the positive Perron-eigenvector associated with $\rho=\rho(D^{-1}_f B_f)$ and $\alpha$ is small enough such that $c=\alpha v\leq 1-\frac{1}{\rho}$.
We use again the map \eqref{eq:map}. For any $z \in \Tilde{Y}$, we have $\hat{T}(z)=T(D^{-1}_f B_f z+D^{-1}_f[z^{\top}B_{f1}z,z^{*\top}B_{f2}z,\cdots]^{\top})\geq T(D^{-1}_f B_f z) \geq T(\alpha \rho v) \geq T(\rho c) \geq c$. Notice that the map is upper bounded by $u$ when $z$ is in the system domain. Using again the Brouwer's fixed-Point Theorem \cite{shapiro2016fixed}, one confirms that there exists an endemic equilibrium $z^*\gg \mathbf{0}$ if $\rho(D^{-1}_f B_f)>1$.
\end{proof}

Now, we show that the endemic equilibrium is unique and globally stable if $\rho(D^{-1}_f B_f)>1$ and the $\beta^w_{ij}$ terms in $B_f$ and all the higher order parameters in all $B_{fi}$ are sufficiently small.

\begin{theorem} (Global stability of the endemic equilibrium)\label{thm:gsee}
Consider the system \eqref{eq::sys_siws_hyp_z}, if Assumptions \ref{ass:xini1}-\ref{ass:irre} hold and  $R_0=\rho(D^{-1}_f B_f)>1$, then the endemic equilibrium is unique and globally exponentially stable if the $\beta^w_{ij}$ terms in $B_f$ and all the higher order parameters in all $B_{fi}$ are sufficiently small.
\end{theorem}

\begin{proof}
We recall that $\Tilde{Y}=\{z=(x,w)^{\top}|c_i\leq x_i\leq 1, c_{j+n}\leq w_j \leq w_{j\max} \}$, where $c=(c_1,\cdots,c_{n+m})^{\top}=\alpha v$, $v$ is the positive Perron-eigenvector associated with $\rho=\rho(D^{-1}_f B_f)$ and $\alpha$ is small enough such that $c=\alpha v\leq 1-\frac{1}{\rho}$. It is easy to check that $\Tilde{Y}$ is a forward invariant set. By using the condition $c_i=\alpha v_i$, it yields that $\dot{z}_i>0$ when $z_i=c_i$.

Then, the idea is to rewrite \eqref{eq::sys_siws_hyp_z} in the form of $\dot{z}=D(z,z^*)(z-z^*)$ and apply Theorem 3.1 (Exponential convergence from Coppel’s inequality) in \cite{cisneros2021multi}. More precisely, we can rewrite \eqref{eq::sys_siws_hyp_z} as
        $\dot{z}=D(z,z^*)(z-z^*),$
where 
{\small\begin{equation*}
\begin{split}
&D(z,z^*):= -D_f+(I-{Z}^*)B_f-\Dg(\left[\begin{matrix}
	B & B_w \\
\mathbf{0}& \mathbf{0}  \\
	\end{matrix}\right] z) \\
 &+(I-{Z}^*)\left[\begin{matrix}
            \frac{1}{2}(z+z^*)^\top B_{f1}+\frac{1}{2}(z+z^*)^{\top} B_{f1}^\top\\
            \vdots \\
           \frac{1}{2}(z+z^*)^\top B_{fn+m}+\frac{1}{2}(z+z^*)^{\top}B_{fn+m}^\top
        \end{matrix}
        \right]\\
        &-\Dg[(z^\top B_{f1}z,\cdots,z^\top B_{fn}z,0,\cdots,0)^\top]
\end{split}
\end{equation*}} is an irreducible Metzler matrix.

Again, we define $\tilde{z}^*=(\frac{x^*}{2},w^*)^\top$.
For $i\geq n+1$, we have
\begin{equation*}\label{eq::D1}
\begin{split}
 (D\tilde{z}^*)_i=-\frac{1}{2}\delta^w_{i-n} w_{i-n}^*+\frac{1}{4}\sum_{j,k} c^{xx}_{(i-n)jk} (x_i^* x_k+x_i x_k^*).
\end{split}
\end{equation*}

On the other hand, for $i\leq n$ we have
{\small\begin{equation*}\label{eq::D2}
\begin{split}
 (D\tilde{z}^*)_i &=(1-x_i^*)\Big(\sum_j \frac{1}{2}\beta^w_{ij}w_j^*\Big)\\
 &-\frac{1}{2}x^*_i\Big(\sum_j \beta_{ij}x_j + \sum_{j}\beta^w_{ij}w_j+\sum_{j,k} \beta_{ijk} x_ix_j\\
 &+\sum_{j,k} \beta_{ijk}^{wx} w_jx_k 
 +\sum_{j,k} \beta_{ijk}^{xw} x_jw_k\Big)
 \\
 &+(1-x_i^*)\Big( \frac{1}{4}\sum_{j,k}\beta^{xw}_{ijk} (x^*_j w_k+ x_j w^*_k+ x^*_j w^*_k) \\
 &+ \sum_{j,k}\frac{1}{4}\beta^{wx}_{ijk}(w^*_j x^*_k+w^*_j x_k+ w_j x^*_k) \\
 &+ \sum_{j,k}\frac{1}{4}\beta^{xx}_{ijk} (x^*_j x_k+x_j x^*_k)  \Big).
\end{split}
\end{equation*}}

One can confirm that if the $\beta^w_{ij}$ terms in $B_f$ and all the higher order parameters in all $B_{fi}$ ($\beta^{xw}_{ijk},\beta^{xx}_{ijk},\beta^{wx}_{ijk},c^{xx}_{ijk}$) are sufficiently small, then $D\tilde{z}^*$ is upper bounded by some negative constant. Thus, according to Theorem 3.1 (Exponential convergence from Coppel’s inequality) in \cite{cisneros2021multi}, we complete the proof.
\end{proof}

Next, we show that the system \eqref{eq::sys_siws_hyp_z} is actually an irreducible monotone system.

\begin{theorem} [Irreducible monotone system]\label{thm:irrsys}
If Assumption \ref{ass:xini1}-\ref{ass:irre} hold, the system \eqref{eq::sys_siws_hyp_z} is an irreducible monotone system in $\mathbf{{D}}$ (defined in Lemma \ref{lem:GP}). Furthermore, the model \eqref{eq::sys_siws_hyp_z} has a finite number of equilibria in $\mathbf{\Bar{D}}$, which is the closure of $\mathbf{{D}}$ (defined in Lemma \ref{lem:GP}), for a generic choice of parameters. Moreover, the solution of \eqref{eq::sys_siws_hyp_z} converges to an equilibrium for almost all initial conditions. That is to say, the set of initial conditions, such that the model does not converge to an equilibrium, is a set of Lebesgue measure zero.
\end{theorem}

\begin{proof}
Firstly, we calculate the Jacobian of the system \eqref{eq::sys_siws_hyp_z}:
\begin{align*}
    J=-D_f+(I-Z)(B_f+J_H(z))-Z\Dg(B_f z+H(z));
 \end{align*}
where $J_H(z)$ is the Jacobian of $H(z)$.
Note that $B_f+J_H(z)> \mathbf{0}$. Furthermore, $D_f$ is diagonal. Thus, the Jacobian is an irreducible Metzler matrix at an arbitrary point in $\mathbf{{D}}$ along the trajectories of \eqref{eq::sys_siws_hyp_z}. This confirms that the system \eqref{eq::sys_siws_hyp_z} is an irreducible monotone system by Lemma 2.2 of \cite{ye2022convergence}.

We just need to further confirm that the system \eqref{eq::sys_siws_hyp_z} has a finite number of equilibria in the closure of $\mathbf{\Bar{D}}$ for a generic choice of parameter. We know that the healthy state is globally exponentially stable if  $\rho(D^{-1}_f B_f+D^{-1}_f[u^{\top}B_{f1},\cdots,u^{\top}B_{fn+m}]^{\top})<1$, which indicates that the healthy state is a unique equilibrium. Another example could be to choose the infection rates (some are zeros, the rest take positive values) and recovery rates (set to be 1) such that it yields a partially-decoupled equation of $\{-1+ (1-x_i)[\beta_{ii}+\beta_{iii}x_i]\}x_i=0$ and $-w_i+ [c_{ij}x_j+c_{ijj}(x_j)^2]=0$. It is straightforward to see that such an equation system has a finite number of solutions. We indeed see some particular examples. Thus, we found out that the model \eqref{eq::sys_siws_hyp_z} has a finite number of equilibria in the closure of $\mathbf{\Bar{D}}$ with a particular choice of parameter. We can also see that the equation $\dot{z}=\mathbf{0}$ is actually a polynomial equation. Thus, the model has a finite number of equilibria in the closure of $\mathbf{\Bar{D}}$ if the parameters do not lie on a certain algebraic set of measure zero by Theorem B.1 and Corollary B.2 in \cite{ye2022convergence}.

Moreover, by Lemma 2.3 of \cite{ye2022convergence}, the solutions of \eqref{eq::sys_siws_hyp_z} converge to an equilibrium with some domain of attraction; by the same lemma, the set of initial conditions, such that the model does not converge to an equilibrium, is a set of measure zero.
\end{proof}

\begin{remark}
As a consequence of \eqref{eq::sys_siws_hyp_z} being an irreducible monotone system and according to the Theorem \ref{thm:irrsys}, the convergence to equilibrium occurs for almost all initial conditions. Once the system processes a unique locally stable equilibrium, it is then globally stable. If the healthy state is unstable, then the convergence to an endemic equilibrium happens for almost all initial conditions. However, there may be multiple endemic equilibria. In the case of multiple locally stable equilibria, then multi-stability is induced. However, we have not determined the domain of attraction of each locally stable equilibrium yet. This remains future work. 
\end{remark}

\section{Extension to a bi-virus competing SIWS model with higher-order interactions}\label{sec:bi}
In most scenarios, not just one virus (or information) spreads over the network. Usually, several different viruses spread simultaneously and compete with each other. Concrete examples of this competition in real life include competing opinions on a social network \cite{sahneh2014competitive} and competing rumors spreading in politics \cite{trpevski2010model}. Here, we study the simplest case of a bi-virus competition as a typical example. Since we concentrate on a case of extreme competition, we assume both viruses are not compatible with each other and one can only get infected with one virus at a time. We want to emphasize that this competitive assumption and coupling method are widely adopted among related works in epidemic models on a conventional graph \cite{liu2019analysis,cui2022discrete}. Then, a bi-virus model can be obtained in at least two ways. One can directly couple two separate single-virus models \eqref{eq::sys_siws_hyp_z} into one system. Alternatively, one can first model a Markov chain for the bi-virus case and repeat the procedure in the section \ref{subsection:GM} and \ref{sec::SIWS}.

 The model \eqref{eq::sys_siws_hyp_z} can be straightforwardly extended into a bi-virus competing SIWS model, based on the assumption of no simultaneous co-infection, which leads to
{
\small\begin{align}
\dot{x}^1_i&= -\delta_i^1 x_i^1 
+\left[ (1-x_i^1-x_i^2) (\sum_{j=1}^{n} \beta^1_{ij} x_j^1
 +\sum_{j=1}^{m} \beta_{ij}^{w1} w_j^1 ) \right]\notag\\
 &+\left[ (1-x_i^1-x_i^2) \left(\sum_{j,k\in \mathbf{N}_3^{\text{i}}} \beta^{xx1}_{ijk} x_i^1 x_j^1
 +\sum_{j,k\in \mathbf{N}_3^{\text{i}}} \beta_{ijk}^{wx1} w_j^1x_k^1 
 +\sum_{j,k\in \mathbf{N}_3^{\text{i}}} \beta_{ijk}^{xw1} x_j^1w_k^1
 \right) \right],
 \label{eq::sys_sis_dt_x1}\\
 \dot{w}_j^1&= -\delta_{j1}^{w}w_{j}^1+\sum_{k=1}^{n}c_{jk}^{w1}x_{k}^1+\sum_{i,k\in \mathbf{N}_3^{\text{i}}}c_{ijk}^{xx1}x^1_{i}x^1_{k}, \label{eq::sys_sis_dt_w1}\\
 \dot{x}^2_i&= -\delta_i^2 x_i^2 
+\left[ (1-x_i^1-x_i^2) (\sum_{j=1}^{n} \beta_{ij}^2 x_j^2
 +\sum_{j=1}^{m} \beta_{ij}^{w2} w_j^2 ) \right]\notag\\
 &+\left[ (1-x_i^1-x_i^2) \left(\sum_{j,k\in \mathbf{N}_3^{\text{i}}} \beta^{xx2}_{ijk} x_i^2 x_j^2
 +\sum_{j,k\in \mathbf{N}_3^{\text{i}}} \beta_{ijk}^{wx2} w_j^2x_k^2 +\sum_{j,k\in \mathbf{N}_3^{\text{i}}} \beta_{ijk}^{xw2} x_j^2w_k^2
 \right) \right],
 \label{eq::sys_sis_dt_x2}\\
 \dot{w}_j^2&= -\delta_{j}^{w2}w_{j}^2+\sum_{k=1}^{n}c_{jk}^{w2}x_{k}^2+\sum_{i,k\in \mathbf{N}_3^{\text{i}}}c_{ijk}^{xx2}x^2_{i}x^2_{k}, \label{eq::sys_sis_dt_w2}
\end{align}
}\par
\noindent where all notations remain the same as Table \ref{tab:notation}  and the superscript $1,2$, denotes the first or the second virus. Without loss of generality and for simplicity, we refer to $\nu=1$ as the first virus and to $\nu=2$ as the second. We recall that, for now, we deal with the interaction of up to 3 bodies.

For $\nu=1,2$, we can rewrite our system \eqref{eq::sys_sis_dt_x1}-\eqref{eq::sys_sis_dt_w2} as
\begin{equation}\label{eq::sys_siws_dt_zbi1}
\begin{split}
    \dot{z}^\nu&=-D_f^\nu z^\nu +(I-Z^1-Z^2)(B_f^\nu z^\nu+H^\nu(z^\nu)).
\end{split}
\end{equation}

In the following, we present the bi-virus version of assumptions \ref{ass:xini1}, \ref{ass:para}, and \ref{ass:irre}.

\begin{assumption}\label{ass:xini2}
The initial condition of infection level satisfies $x^{\nu}_i(0)\in [0,1]$ for all for all $\nu=1,2$ and $i=1,\ldots,n$. Furthermore, $x^1_i(0)+x^2_i(0) \leq 1$ holds for all $i=1,\ldots,n$.
\end{assumption}

\begin{remark}
We can now obtain the upper bound of each pathogen dynamic as 
\begin{equation}\label{eq:wkjmax}
    w^\nu_{j\max}=\frac{\sum_{k=1}^{n}c_{jk}^{w\nu}+\sum_{i,k\in \mathbf{N}_3^{\text{i}}}c_{jik}^{xx\nu}}{\delta_j^w}.
\end{equation}
\end{remark}

\begin{assumption}\label{ass:para2} 
The matrices $B_f^\nu$, $B_{fj}^\nu$ are non-negative for $\nu=1,2$. The diagonal entries of $D_f^\nu$ are positive for $\nu=1,2$.
\end{assumption}

\begin{assumption}\label{ass:irre2}
The matrix $B^\nu_f$ is irreducible for $\nu=1,2$.
\end{assumption}

For simplicity, we define the notation $F(z)^\nu=B_f^\nu z^\nu$, $f_i^\nu=(B_f^\nu z^\nu)_i$, $h_i^\nu=(H^\nu)_i$.  Then, the concrete system \eqref{eq::sys_siws_dt_zbi1} is a special case of the general system \eqref{eq::sys_siws_dt_zbi}, which is discussed later. 

\begin{lemma}[General properties]\label{lem:2a}
Consider the bi-virus model \eqref{eq::sys_siws_dt_zbi1}. Let $w^k_{j\max}$ be as defined in \eqref{eq:wkjmax},
\begin{equation}
\begin{split}
    \mathbf{D}&=\left\{z|z=(x^1,w^1,x^2,w^2)^{\top}=(x^1_1,\cdots,x^1_n,w^1_1,\cdots,w^1_m,x^2_1,\cdots,x^2_n,w^2_1,\cdots,w^2_m)^\top, \right.\\
    &\qquad \left. x^\nu_i \in (0,1), \, w^\nu_j \in (0,w^\nu_{j\max}), \, \nu=1,2 \right\},
\end{split}
\end{equation}
and $\mathbf{\Bar{D}}$ be its closure. If Assumptions \ref{ass:xini2}-\ref{ass:para2} hold, then
\begin{itemize}
    \item [i)] the set $\mathbf{\Bar{D}}$ is positively invariant;
    \item [ii)] the set $\mathbf{D}$ is also positively invariant;
    \item [iii)] the origin is always an equilibrium of the model;
    \item [iv)] the dynamics of each virus in \eqref{eq::sys_siws_dt_zbi1} are upper bounded by its single virus counterpart in \eqref{eq::sys_siws_hyp_z} with the same system parameters.
\end{itemize}
\end{lemma}

\begin{proof}
This lemma is a special case of Lemma \ref{lem:2}. Thus, the proof is omitted here and we provide the proof of a more general case in Lemma \ref{lem:2}.
\end{proof}

\begin{remark}
    Notice that now we study a bi-virus system \eqref{eq::sys_siws_dt_zbi1}. Hence, the system domain $\mathbf{D}$ is different from the system domain of a single-virus system in the former section \eqref{eq::sys_siws_hyp_z}. We keep the same notation because we are referring to the equivalent physical domain.
\end{remark}

\begin{theorem}[Healthy-state behaviour] \label{thm:healthybi}
Consider the system \eqref{eq::sys_siws_dt_zbi1}. If Assumptions \ref{ass:xini2}-\ref{ass:irre2} hold, then the healthy state (zero equilibrium) always exists and is
\begin{enumerate}
    \item  locally stable if $\rho((D^{\nu}_f)^{-1} B_f^\nu)<1$ for $\nu=1,2$,
    \item  globally exponentially stable if $$\rho((D^\nu_f)^{-1} B^\nu_f+(D^\nu_f)^{-1}[u^{\nu\top}B^\nu_{f1},\cdots,u^{\nu\top}B^\nu_{fn+m}]^{\top})<1,$$ with  ${u}^\nu=[\mathbf{1}^n,w^\nu_{1\max},\cdots,w^\nu_{m\max}]^{\top}$, and for $\nu=1,2$,
    \item unstable if $\rho((D^{\nu}_f)^{-1} B_f^\nu)>1$ for either $\nu=1$ or $\nu=2$.
\end{enumerate}
\end{theorem}

\begin{proof}
It is clear that the healthy state (zero equilibrium) always exists. The Jacobian matrix of the model at the origin reads as
 \begin{equation*}
     J(0,0)=\left[ 
        \begin{matrix}
	-D^1_f+B_f^1& \mathbf{0}  \\
	\mathbf{0} & -D^2_f+B_f^2   \\
	\end{matrix}
        \right].
 \end{equation*}
 
Thus, the healthy state is locally stable if $\rho((D^{\nu}_f)^{-1} B_f^\nu)<1$ for $\nu=1,2$ and unstable if $\rho((D^{\nu}_f)^{-1} B_f^\nu)>1$ for either $\nu=1$ or $\nu=2$.

The global exponential stability when  $$\rho((D^\nu_f)^{-1} B^\nu_f+(D^\nu_f)^{-1}[u^{\nu\top}B^\nu_{f1},\cdots,u^{\nu\top}B^\nu_{fn+m}]^{\top})<1$$ for $\nu=1,2$, is a direct consequence of statement iv) in Lemma \ref{lem:2a}, Theorem \ref{thm:healthysig} and the comparison principle.
\end{proof}

\begin{theorem} [Dominant endemic equilibrium]\label{thm:deq}
Consider the bi-virus model \eqref{eq::sys_siws_dt_zbi1}, if Assumptions \ref{ass:xini2}-\ref{ass:irre2} hold, we have the following statements. If $z^*=(x^*,w^*)^\top$ is an endemic equilibrium of the single-virus counterpart \eqref{eq::sys_siws_hyp_z} of the first virus, then $(z^*,0,0)^{\top}$ is the dominant endemic equilibrium of \eqref{eq::sys_siws_dt_zbi1}. The dominant endemic equilibrium is locally stable when $z^*$ is locally stable and $\rho(-D^2_{f}+(I-Z^*)B_f^2)<0$ with respect to the single-virus system and is unstable when $z^*$ is unstable or $\rho(-D^2_{f}+(I-Z^*)B_f^2)>0$, where $Z^*=\left[ 
        \begin{matrix}
	x^* & 0  \\
	0 & 0  \\
	\end{matrix}
        \right]$. Furthermore, if $\rho((D^{1}_f)^{-1} B^1_f)>1$, $\rho((D^2_f)^{-1} B^2_f+(D^2_f)^{-1}[u^{2\top}B_{f1},\cdots,u^{2\top}B_{fn+m}]^{\top})<1$, the $\beta^{w1}_{ij}$ terms in $B_f$ and all the higher order parameters in all $B^1_{fi}$ are sufficiently small, then the dominant endemic equilibrium is globally stable with domain of attraction $\Bar{\mathbf{D}}\setminus \mathcal{H}$, where $\mathcal{H}=\{[z^1, z^2]^{\top}\,|\, z^2\in\Bar{\mathbf{D}}, z^1=\mathbf{0}\}$.
\end{theorem}

\begin{proof}
If we set $z^2=\mathbf{0}$, the system \eqref{eq::sys_siws_dt_zbi1} is reduced into the single virus counterpart of the first virus. Thus, if $z^*$ is an endemic equilibrium of the single-virus counterpart \eqref{eq::sys_siws_hyp_z} of the first virus, then $(z^*,0,0)^{\top}$ is the dominant endemic equilibrium of the system \eqref{eq::sys_siws_dt_zbi1}.

The Jacobian matrix of the model at the dominant endemic equilibrium is 
 \begin{equation*}
     J(z^*,\mathbf{0})=\left[ 
        \begin{matrix}
	\hat{J}(z^*) & \mathbf{Res}  \\
	\mathbf{0} & -D^2_{f}+(I-Z^*)B_f^2   \\
	\end{matrix}
        \right],
 \end{equation*}
 where $\hat{J}(z^*)$ is the Jacobian of the single-virus counterpart at the endemic equilibrium $z^*$ and $\mathbf{Res}$ is the matrix of some residual which we do not need to deal with.
 
 By observing the Jacobian $J(z^*,\mathbf{0})$ and notice that $\hat{J}(z^*)$ and $-D^2_{f}+(I-Z^*)B_f^2$ is irreducible nonnegative, we can directly obtain the results of local stability. That is, the dominant endemic equilibrium is locally stable when $z^*$ is locally stable and $\rho(-D^2_{f}+(I-Z^*)B_f^2)<0$ with respect to the single-virus system and is unstable when $z^*$ is unstable or $\rho(-D^2_{f}+(I-Z^*)B_f^2)>0$. 
 
 For the case of global stability, since $$\rho((D^2_f)^{-1} B^2_f+(D^2_f)^{-1}[u^{2\top}B_{f1},\cdots,u^{2\top}B_{fn+m}]^{\top})<1$$ holds, the single-virus counterpart of the second virus converges to the healthy state. By the comparison principle, the dynamics of the second virus converge to zero. The dynamic of the first virus is thus its single virus counterpart plus a vanishing perturbation. Therefore, the whole system converges to the dominant endemic equilibrium.
\end{proof}

\begin{remark}
The analytical results of the Theorem \ref{thm:deq} are still true if we interchange the superscript $1$ (first virus) and $2$ (second virus). Bistability may also occur in the bi-virus system because the condition for the local stability of the healthy state and a dominant endemic equilibrium may be satisfied at the same time.
\end{remark}

\begin{theorem} [Coexisting equilibrium]\label{thm:coexist}
Consider the model \eqref{eq::sys_siws_dt_zbi1}, if Assumptions \ref{ass:xini2}-\ref{ass:irre2} hold,  we have the following statements. 
If $\rho((D^{\nu}_f)^{-1} B^\nu_f)>1$ for all $\nu=1,2$ and $\rho(-D^1_{f}+(I-Z^{2*})B^1_f)>0$ as well as $\rho(-D^2_{f}+(I-Z^{1*})B^2_f)>0$, where $z^{1*}$ and $z^{2*}$ are the endemic equilibrium of the single-virus counterpart of the first and second virus respectively and the notation $Z^{1*}$ and $Z^{2*}$ sticks to the form of \eqref{eq::notations1}, then there exists at least one coexisting equilibrium $(\hat{z}^1,\hat{z}^2)\gg \mathbf{0}$ such that $\hat{x}^1+\hat{x}^2\leq \mathbf{1}$, where $\hat{z}^\nu=(\hat{x}^\nu,\hat{w}^\nu)^\top$.
\end{theorem}

\begin{proof}
Let us define the map $M: \mathbb{R}^{2n+2m}\rightarrow \mathbb{R}^{2n+2m}$ as follows: component-wise, $M$ is given by
\begin{equation}
\begin{aligned}
M^1\left(z^1, z^2\right)_i &=\frac{\left(1-z_i^2\right)\left(\left(D_f^1\right)^{-1} B_f^1 z^1+\left(D_f^1\right)^{-1} H^1(z^1)\right)_i}{1+\left(\left(D_f^1\right)^{-1} B_f^1 z^1 +\left(D_f^1\right)^{-1} H^1(z^1) \right)_i} \\
M^2\left(z^1, z^2\right)_i &=\frac{\left(1-z_i^1\right)\left(\left(D_f^2\right)^{-1} B_f^2 z^2+\left(D_f^2\right)^{-1} H^2(z^2)\right)_i}{1+\left(\left(D_f^2\right)^{-1} B_f^2 z^2+\left(D_f^2\right)^{-1} H^2(z^2)\right)_i},
\end{aligned}
\end{equation}
for $i\leq n$. Furthermore, for $i\geq n+1$, we let
\begin{equation}
\begin{aligned}
M^1\left(z^1, z^2\right)_i &=(D_f^1)^{-1} (B_f^1 z^1+H^1(z^1))_i, \\
M^2\left(z^1, z^2\right)_i &=(D_f^2)^{-1} (B_f^2 z^2+H^2(z^2))_i.
\end{aligned}
\end{equation}

It is straightforward to check that the fixed point of the map $M(z)=M(z^1,z^2)=(M^1(z^1,z^2),M^2(z^1,z^2))^\top$ yields the equilibrium point of the model \eqref{eq::sys_siws_dt_zbi1}. Furthermore, the map $M^1$ is a non-increasing function of $z^2$ and a non-decreasing function of $z^1$. Similarly, the map $M^2$ is a non-decreasing function of $z^2$ and a non-increasing function of $z^1$. 

Suppose $z^{1*}$ and $z^{2*}$ are the endemic equilibrium of the single-virus counterpart \eqref{eq::sys_siws_hyp_z} of the first and second virus respectively. It follows that $M_1(z^{1*},\mathbf{0})=z^{1*}$ and $M_2(\mathbf{0},z^{2*})=z^{2*}$. Then, we look at the map $M_1$. If $z^2\geq \mathbf{0}$, then it holds $M_1(z^{1*},z^2)\leq z^{1*}$. Moreover, if $z^1\leq z^{1*}$, then $M_1(z^1,z^2)\leq z^{1*}$. Similarly, we have $M_2(z^1,z^2)\leq z^{2*}$ if $z^1\geq \mathbf{0}$ and $z^2\leq z^{2*}$. Combining these two results, we get $M(z^1,z^2)\leq M(z^{1*},z^{2*})$ if $(z^1,z^2)\leq(z^{1*},z^{2*})$.

Now, we construct another map $\hat{M}$.  Component-wise, the expression for each row is given by,
\begin{equation}
\begin{aligned}
\hat{M}^1\left(z^1, z^2\right)_i &=\frac{\left(1-z_i^2\right)\left(\left(D_f^1\right)^{-1} B_f^1 z^1\right)_i}{1+\left(\left(D_f^1\right)^{-1} B_f^1 z^1  \right)_i}\leq M^1\left(z^1, z^2\right)_i \\
\hat{M}^2\left(z^1, z^2\right)_i &=\frac{\left(1-z_i^1\right)\left(\left(D_f^2\right)^{-1} B_f^2 z^2\right)_i}{1+\left(\left(D_f^2\right)^{-1} B_f^2 z^2\right)_i}\leq M^2\left(z^1, z^2\right)_i;
\end{aligned}
\end{equation}
for $i\leq n$. These inequalities are true because the function $\frac{s}{1+s}$ is increasing as $s$ increases. Then, similar proof techniques from Theorem 5 of \cite{janson2020networked} can be further applied.

Then, for $i\geq n+1$,
\begin{equation}
\begin{aligned}
\hat{M}^1\left(z^1, z^2\right)_i &=(D_f^1)^{-1} (B_f^1 z^1)\leq M^1\left(z^1, z^2\right)_i \\
\hat{M}^2\left(z^1, z^2\right)_i &=(D_f^2)^{-1} (B_f^2 z^2)\leq M^2\left(z^1, z^2\right)_i.
\end{aligned}
\end{equation}
Hence, it follows that $\hat{M}(z)\leq M(z)$. Due to the condition $\rho(-D^1_{f}+(I-Z^{2*})B^1_f)>0$ as well as $\rho(-D^2_{f}+(I-Z^{1*})B^2_f)>0$ and also considering that $-D^1_{f}+(I-Z^{2*})B^1_f$ and $-D^2_{f}+(I-Z^{1*})B^2_f$ are both irreducible Metzler matrices, we can obtain $\rho((I-Z_{3-\nu}^*)(D^\nu_{f})^{-1}B^\nu_f)>1$, for $\nu=1,2$. Note that $(I-Z_{3-\nu}^*)(D^\nu_{f})^{-1}B^\nu_f$ is irreducible nonnegative. Let $\lambda^\nu>1$ and $y^\nu$ be its corresponding simple eigenvalue and positive eigenvector. Since $(D^\nu_{f})^{-1}B^\nu_f$ is also irreducible nonnegative, we have $((D^\nu_{f})^{-1}B^\nu_f y^\nu)_k>0$. Hence, there must exist $\epsilon^1>0$ and $\epsilon^2>0$ such that
\begin{equation}
\begin{array}{r}
\epsilon^1<\min \left\{\frac{\lambda^1-1}{\max_{i \in[n+1]}\left(\left(D_f^1\right)^{-1} B_f^1 {{y}}^1\right)_i}, \min_{i \in[n+1]} \frac{({z}_1^*)_i}{{{y}}_i^1}\right\} \\
\epsilon^2<\min \left\{\frac{\lambda^2-1}{\max _{i \in[n+1]}\left(\left(D_f^2\right)^{-1} B_f^2 {{y}}^2\right)_i}, \min _{i \in[n+1]} \frac{({z}_2^*)_i}{{{y}}_i^2}\right\}
\end{array}
\end{equation}
It follows that 
\begin{equation}
\begin{aligned}
&1+\max _{i \in[n]}\left(\left(D_f^1\right)^{-1} B_f^1 \epsilon^1 {{y}}^1\right)_i<\lambda^1 \\
&1+\max _{i \in[n]}\left(\left(D_f^2\right)^{-1} B_f^2 \epsilon^2 {{y}}^2\right)_i<\lambda^2
\end{aligned}
\end{equation}
Thus, for $i\leq n$,
\begin{equation*}
\begin{aligned}
\hat{M}_i^1\left(\epsilon^1 {{y}}^1, z^{2*}\right) &=\frac{\left(\left(I-Z^{2*}\right)\left(D_f^1\right)^{-1} B_f^1 \epsilon^1 y^1\right)_i}{1+\left(\left(D_f^1\right)^{-1} B_f^1 \epsilon^1 {y}^1\right)_i} \\
&=\frac{\lambda^1 \epsilon^1 {y}_i^1}{1+\left(\left(D_f^1\right)^{-1} B_f^1 \epsilon^1 {{y}}^1\right)_i}>\epsilon^1 {y}_i^1 
\end{aligned}
\end{equation*}

Similarly, we have $\hat{M}_i^2\left({z}^*_1, \epsilon^2 {{y}}^2\right) >\epsilon^2 {y}_i^2$ for $i\leq n$. On the other hand, for $i\geq n+1$, 
\begin{equation*}
\begin{aligned}
\hat{M}_i^1\left(\epsilon^1 {{y}}^1, z^{2*}\right) &=\left((D_f^1)^{-1} (B_f^1 \epsilon^1 {{y}}^1)\right)_i\\
&=\epsilon^1 \lambda^1 {y}_i^1 >\epsilon^1 {y}_i^1, 
\end{aligned}
\end{equation*}
and analogously, we have $\hat{M}_i^2\left(z^{1*}, \epsilon^2 {{y}}^2 \right) >\epsilon^2 {y}_i^2$ for $i\geq n+1$.

By using the monotonicity of the map ($M,\hat{M}$) and noticing that $M\geq \hat{M}$, it follows that
\begin{equation*}
{M}^1\left(z^1, z^2\right)\geq\hat{M}^1\left(z^1, z^2\right)>\epsilon^1 {{y}}^1,\quad M^2\left(z^1, z^2\right)\geq \hat{M}^2\left(z^1, z^2\right)>\epsilon^2 {{y}}^2,
\end{equation*}
if $ \epsilon^1 {{y}}^1 \leq z^1 \leq z^{1*}$ and $ \epsilon^2 {{y}}^1 \leq z^2 \leq z^{2*}$. Thus, the map $M$ maps $\{ (z^1,z^2)| \epsilon^1 {{y}}^1 \leq z^1 \leq z^{1*},\quad \epsilon^2 {{y}}^1 \leq z^2 \leq z^{2*}\}$ into itself. According to the Brouwer fixed point theorem \cite{shapiro2016fixed}, we finish the proof.
\end{proof}

\begin{remark}
In \cite{li2022competing}, the authors study a very similar mean-field bi-virus planar SIS model in a simplicial complex. Such a model can be seen as a simplified version of our bi-virus model \eqref{eq::sys_siws_dt_zbi1}, as it does not account for the indirect spreading via resources. In \cite{li2022competing}, it is shown that the mean-field bi-virus model also possesses some coexisting equilibria. Further simulation results suggest that the coexistence equilibrium is unstable. In contrast, we have provided rigorous proof of the existence of a coexisting equilibrium in an even more general network structure and modeling setting. Its stability remains a future research question. However, from the numerical study performed by us, such a coexisting equilibrium seems to always be unstable.
\end{remark}

\begin{theorem} [Irreducible monotone system]\label{thm:irsbi}
If Assumptions \ref{ass:xini2}-\ref{ass:irre2} hold, the bi-virus model \eqref{eq::sys_siws_dt_zbi1} is an $(\mathbf{0}_{n+m},\mathbf{1}_{n+m})-$type irreducible monotone system in $\mathbf{{D}}$. The model has a finite number of equilibria in $\mathbf{{\bar D}}$ for a generic choice of parameters. The model converges to a equilibrium for almost all initial conditions. That is to say, the set of the initial conditions, such that the model does not converge to an equilibrium, is a set of Lebesgue measure zero.
\end{theorem}

\begin{proof}
We recall that the system \eqref{eq::sys_siws_dt_zbi1} can be formulated as
\begin{equation*}
\begin{split}
    \dot{z}^\nu&=-D_f^\nu z^\nu +(I-Z^1-Z^2)(B_f^\nu z^\nu+H^\nu(z^\nu)), \nu=1,2.
\end{split}
\end{equation*}
Thus, the Jacobian is 
\begin{equation*}
    J=\left[ 
        \begin{matrix}
	J^1 & A^1  \\
	A^2 & J^2   \\
	\end{matrix}
        \right];
\end{equation*}
where
\begin{equation*}
\begin{split}
    J^\nu&=-D^\nu_f+(I-{Z}^1-Z^2) 
    (B^\nu_f+[z^{\nu\top}(B^\nu_{f1}+B_{f1}^{\nu\top}),\cdots,z^{\nu\top}(B^\nu_{fn+m}+B_{fn+m}^{\nu\top})]^\top)\\
    &\qquad-\hat{I}([z^{\nu\top} B^\nu_{f1} z,\cdots,z^{\nu\top} B^\nu_{fn+m} z^\nu]^\top+B_f^\nu z^\nu)
\end{split}
\end{equation*}
and $A^\nu= -(B_f^\nu z^\nu+H(z^\nu))$. We notice that $J^\nu$ is an irreducible Metzler matrix and $A^\nu$ is a non-positive matrix as long as the system variables are in  $\mathbf{\Bar{D}}$. Let $P=\Dg((1,\cdots,1,-1,\cdots,-1)^\top)$, then $PJP$ is an irreducible Metzler matrix. Thus, the bi-virus model \eqref{eq::sys_siws_dt_zbi1} is an $(\mathbf{0}_{n+m},\mathbf{1}_{n+m})-$type irreducible monotone system in $\mathbf{{D}}$.

Furthermore, by Theorem B.1 and Corollary B.2 in \cite{ye2022convergence}, the system \eqref{eq::sys_siws_dt_zbi1} has a finite number of equilibria in the closure of $\mathbf{{D}}$ if the parameters do not lie on a certain algebraic set of measure zero because we can observe that there are some cases where the healthy state is the unique equilibrium and is globally stable and isolated. Alternatively, we can set parameters so that the equation system is partially decoupled, similar to the proof of Theorem \ref{thm:irrsys}. We further see it yields a finite number of soultions.

By Lemma 2.3 of \cite{ye2022convergence}, the solutions of \eqref{eq::sys_siws_dt_zbi1} converge to an equilibrium with some domain of attraction; By the same Lemma, the set of the initial conditions, such that the model does not converge to an equilibrium, is a set of measure zero.
\end{proof}

\section{A general SIWS process on a hypergraph}\label{sec:abstract}
Recall that the single-virus system \eqref{eq::sys_siws_hyp_z} and the bi-virus system \eqref{eq::sys_siws_dt_zbi1} studied so far are based on a polynomial interaction function. In reality, the interaction functions can be any general smooth function. Thus, in this section, we propose a generalized SIWS model, where interaction functions remain abstract. Note also that we only deal with up to 3-body interactions with polynomial functions in sections \ref{sec:model}-\ref{sec:bi}. The results in this section will be also valid for a SIWS model beyond 3-body interactions with either polynomial or non-polynomial but smooth functions.

Let us now consider the general SIWS model on a hypergraph of $n$ populations and $m$ resources. Component-wise, we propose a system given by
\begin{align}
    \dot{z}_i&=\dot{x}_i = -\delta_i x_i 
    + (1-x_i)  \left( f_i(z) 
 +  h_i(z) \right),  \qquad 1\leq i \leq n; \label{eq::sys_siws_hyp_xi}\\
 \dot{z}_i &=\dot{w}_{j} = -\delta^{w}_{j} w_j+ f_i(x) 
 +  h_i(x),\qquad \qquad n+1\leq i=j+n \leq n+m, \label{eq::sys_siws_hyp_wi} 
\end{align}
where the system variable is $z=(x,w)^\top\in\mathbb R^{n+m}$, $f_i(z), i=1,\cdots,n$ denotes the pairwise interaction $f_i(x,w)=\sum_{j=1}^{n+m} A_{ij} f_{ij}(z_j)$ on a conventional graph, $h_i(z)$ denotes the higher-order interaction $$h_i(z)=\sum_{j,k=1}^{n+m} A_{ijk} f_{ijk}(z_j,z_k)+ \sum_{j,k,l=1}^{n+m} A_{ijkl} f_{ijkl}(z_j,z_k,z_l)+\cdots$$ among those higher-order hyperedges on a hypergrpah. For the network structure, we assume that the resources are independent of each other and each resource is connected with at least one population node. For this reason, the interaction function in $w_j$ is reduced to $f_i(x) +  h_i(x)$, where $f_i,h_i: \mathbb{R}^n\rightarrow \mathbb{R}$, instead of $f_i(z) +  h_i(z)$, where $f_i,h_i: \mathbb{R}^{n+m}\rightarrow \mathbb{R}$. The interpretation of this model is that we take further heterogeneity into account and consider different kinds of interactions (not necessarily to be a polynomial).

We can further compactify \eqref{eq::sys_siws_hyp_xi}-\eqref{eq::sys_siws_hyp_wi} as 
\begin{align}
    \dot{x}&= -D x 
    + (I-\Dg(x))  \left( F_x(z) 
 +  H_x(z) \right), \label{eq::sys_siws_hyp_x}
 \\
     \dot{w}&= -D_w w+ F_w(z) 
 +  H_w(x), \label{eq::sys_siws_hyp_w}
 \end{align}
where 
\begin{equation}\label{eq:notation1}
    \begin{split}
        D &=\Dg((\delta_1,\cdots,\delta_n)^\top)\\
        D_w &=\Dg((\delta^w_1,\cdots,\delta^w_m)^\top)\\
        F_x&=(f_1(z),\cdots,f_n(z))^\top\\
        F_w&=(f_{n+1}(z),\cdots,f_{n+m}(z))^\top\\
        H_x&=(h_1(z),\cdots,h_n(z))^\top\\
        F_w&=(h_{n+1}(z),\cdots,h_{n+m}(z))^\top
    \end{split}
\end{equation}

We can rewrite our system \eqref{eq::sys_siws_hyp_x}-\eqref{eq::sys_siws_hyp_w} as
\begin{equation}\label{eq::sys_siws_hyp_z2}
    \dot{z}=-D_f z+(I-Z)F(z)+(I-Z)H(z);
\end{equation}
where 
\begin{equation}\label{eq:notation2}
  \begin{aligned}
    &    D_{f}:=\left[ 
        \begin{matrix}
	D & \mathbf{0}  \\
	\mathbf{0} & D^{w}  \\
	\end{matrix}
        \right],
     Z=\left[\begin{matrix}
	\Dg(x) & 0 \\
0& 0  \\
	\end{matrix}\right],
	F(z)=(F_x(z),F_w(z))^\top,\\
&	H(z)=(H_x(z),H_w(x))^\top.
  \end{aligned}
\end{equation}
Now, we present the analytical results of the general SIWS on a hypergraph \eqref{eq::sys_siws_hyp_z2}. Firstly, we need the following assumptions to make sure our model is well-defined.

\begin{assumption}\label{ass:xini}[Natural bound of probability]
The initial condition of infection level satisfies $x_i(0)\in [0,1]$ for all $i=1,\cdots,n$.
\end{assumption}

\begin{assumption}\label{ass:f0} [Healthy state as a fixed point]
The interaction process must satisfy
$F(0)=0$ and $H(0)=0$.
\end{assumption}

\begin{assumption}\label{ass:monotone}[Infected neighbours increases the chance of infection]
The interaction function satisfies $\frac{\partial f_i(z)}{\partial z_j}> 0 $ for all $i$ and $j$ such that there is a pairwise connection $A_{ij}>0$ (conventional edge) and $\frac{\partial h_i(z)}{\partial z_j} \geq 0 $ for all $i=1,
\cdots,n+m$ and $j=1,\cdots,n$.
\end{assumption}

\begin{assumption}\label{ass:ho}[Higher order interaction is in higher order]
The higher-oder interaction function satisfies $\frac{\partial h_i(\mathbf{0})}{\partial z_j} = 0$ for all $i$ and $j$.
\end{assumption}

\begin{assumption}\label{ass:irre3}[Strong connectivity]
The matrix $[A_{ij}]^{(n+m)\times (n+m)}$ is \\ irreducible.
\end{assumption}

In the following Lemma, we describe the general properties of the model \eqref{eq::sys_siws_hyp_z2}.
\begin{lemma}[General properties]\label{lem:gpa}
Consider the model \eqref{eq::sys_siws_hyp_z2}. Let $w_{j\max}=\frac{f_j(1)+h_j(1)}{\delta_j^w}$, $$\mathbf{D}=\left\{z=(x,w)^\top=(x_1,\ldots,x_n,w_1,\ldots,w_m)^\top\,|\, x_i\in(0,1),\,w_j\in(0,w_{j\max}) \right\},$$ and $\mathbf{\Bar{D}}$ be its closure. If Assumptions \ref{ass:f0}-\ref{ass:ho} hold, then
\begin{itemize}
    \item [i)] the set $\mathbf{\Bar{D}}$ is positively invariant;
    \item [ii)] the set $\mathbf{D}$ is also positively invariant;
    \item [iii)] the origin is always an equilibrium of the model and there is no other equilibrium on the boundary of $\mathbf{D}$.
\end{itemize}
\end{lemma}


\begin{proof}
According to Assumptions \ref{ass:f0}-\ref{ass:ho}, we have the following arguments. If $x_i=0$, we have $\dot{x}_i=  f_i(z) +  h_i(z)  \geq 0$. If $w_j=0$, $\dot{w}_j=f_{j+n}(x) +  h_{j+n}(x)\geq 0$. On the boundary (i.e., at least one $z_i=0$), we have that $x=\mathbf{0}$ and $w=\mathbf{0}$ if and only if $\dot{x}= \mathbf{0}$ and $\dot{w}= \mathbf{0}$.

If $x_i=1$, we get $\dot{x}_i=-\delta_i<0$. If $w_j=w_{j\max}$, we have $\dot{w}_j=-f_{n+j}(\mathbf{1})-h_{n+j}(\mathbf{1})+ f_{n+j}(x)+h_{n+j}(x)\leq 0$, since the pairwise interaction is increasing in $x$ and the higher-order interaction function is non-decreasing in $x$. If $w_j=w_{j\max}$ and $x\neq \mathbf{1}$, $\dot{w}_j< 0$. 

All these ensure that statements i) ii) and iii) must be true.
\end{proof}

\begin{remark}
These general properties show that the system is well-defined. The set $\mathbf{\Bar{D}}$ is the system domain of the model \eqref{eq::sys_siws_hyp_z2}.
\end{remark}

It is obvious that the system has an all-zero equilibrium, which is the so-called healthy state. Let us define the notation $J_{F}$, which denotes the Jacobian of the corresponding function $F$.

\begin{theorem}[Healthy-state behaviour] \label{thm:healthy}
Consider the model \eqref{eq::sys_siws_hyp_z2}, if Assumptions \ref{ass:xini}-\ref{ass:irre3} hold, we have the following statements.
The healthy state (zero equilibrium) 
\begin{itemize}
    \item [1)] is locally stable if $\rho(D^{-1}_f J_F(0))<1$;
    \item [2)] is globally exponentially stable if $D^{-1}_f F(z)+D^{-1}_f H(z)\leq dz$ for some $0< d < 1$;
    \item [3)] is unstable if $\rho(D^{-1}_f J_F(0))>1$.
\end{itemize}
\end{theorem}
 \begin{proof}
 The Jacobian matrix of the model at the origin reads 
 \begin{equation*}
     J(0,0) =-D_f+J_F(0).
 \end{equation*}
 Since the matrix $J_F(0)$ is irreducible nonnegative from Assumption \ref{ass:irre3} and $-D_f$ is a negative diagonal matrix, $\rho(D^{-1}_f J_F(0))>1$ is equivalent to 
 $s(-D_f+ B_f)>0$ and $\rho(D^{-1}_f B_f)<1$ is equivalent to 
 $s(-D_f+ B_f)<0$ by Proposition 1 of \cite{liu2019analysis}. This ensures that the origin is locally stable if $\rho(D^{-1}_f J_F(0))<1$ and it is unstable if $\rho(D^{-1}_f J_F(0))>1$.
 
Next, we show global stability. Define the Lyapunov function $V=D_f^{-1}z$. We observe that $V\geq \mathbf{0}$ and $V=\mathbf{0}$ if and only if $z=\mathbf{0}$.
 
 Then, for some $0<d<1$, we get
 \begin{equation}\label{eq:v2}
     \begin{split}
         \dot{V}&=D_f^{-1}\left(-D_f z+(I-Z)F(z)+(I-Z)H(z)\right)\\
         &\leq - z + D_f^{-1}(F(z)+H(z))\\
         &\leq - z + D_f^{-1}dz= (-1+d)  z\\
         & = (-1+d)   D_f D_f^{-1} z
         = (-1+d) \min(\delta, \delta^w)V,
     \end{split}
 \end{equation}
The rest of the proof is similar to the proof of Theorem \ref{thm:healthysig}.
 \end{proof}
 
The condition $D^{-1}_f F(z)+D^{-1}_f H(z)\leq dz$ generally requires that the infection interaction process doesn't increase sharply and is upper bounded by a line with the slope smaller than $1$. Next, we deal with the existence of an endemic equilibrium.

\begin{theorem} [Existence of an endemic equilibrium] \label{thm:end1}
Consider the system \eqref{eq::sys_siws_hyp_z2}, if Assumptions \ref{ass:xini}-\ref{ass:irre3} hold, there exists an endemic equilibrium $z^*\gg \mathbf{0}$ such that $x_i^*\geq \frac{k-1}{k}$ and $w^*_i\geq \frac{(k-1)w_{i\max}}{k}$ for any $i$ such that $B_{fi}\neq \mathbf{0}$ if $\theta=\min_{i} (D^{-1}_f F(\frac{k-1}{k}\tilde{u})+ D^{-1}_f H(\frac{k-1}{k}\tilde{u}))_i \geq \frac{(k-1)^2}{k}\Tilde{u}$, where 
\begin{equation*}
       \tilde{u}_i = \begin{cases}
                        1,  & H(z)_i\not\equiv 0,\qquad i=1,\ldots,n \\
                        w_{i\max}, & H(z)_i\not\equiv 0,\qquad i=n+1,\ldots,n+m\\
                        0,  & H(z)_i \equiv 0
    \end{cases}
\end{equation*}
\end{theorem}
 \begin{proof}
 We use the same map \eqref{eq:map} and it follows the same monotone property. It is easy to check that the fixed point of the map \eqref{eq:map} is also the equilibrium of \eqref{eq::sys_siws_hyp_z2}.
Now, we define the set $Y=\{z=(x,w)^\top\in \mathbb{R}^{n+m}|\frac{1}{2} \tilde u\leq z \leq u\}$.

We have that
\begin{equation*}
\begin{split}
    \hat{T}(z)&=T(D^{-1}_f F(z) +D^{-1}_f H(z))\\
    &\geq T\left(D^{-1}_f F(\frac{1}{2} \tilde u) +D^{-1}_f H(\frac{1}{2}\tilde u)\right)\\
    &\geq T\left(\theta \right)\geq \frac{k-1}{k} \tilde u
\end{split}
\end{equation*}

Notice that $\frac{p_i}{1+p_i}\leq 1$ for all $i\leq n$ and $p_{j+n}\leq w_{j\max}$ for all $j\leq m$. Since $\hat{T}(z)$ is a continuous map and it maps the set $Y$ to itself. The Brouwer Fixed-Point Theorem \cite{shapiro2016fixed} ensures that there exists at least one $z^*\gg \mathbf{0}$ of the given form.
 \end{proof}

 \begin{theorem}[Further results on the existence of an endemic equilibrium]\label{thm:end2ab}
Consider the model \eqref{eq::sys_siws_hyp_z2}, if Assumptions \ref{ass:xini}-\ref{ass:irre3} hold and there exists a $c\neq0$ such that $c\ll u$ and $D^{-1}_f F(c)\geq c$, then there exists an endemic equilibrium $z^*\gg \mathbf{0}$.
\end{theorem}

\begin{proof}
Define the set $\Tilde{Y}=\{z=(x,w)^{\top}|c_i\leq x_i\leq 1, c_{j+n}\leq w_j \leq w_{j\max} \}$, where $c=(c_1,\cdots,c_{m+n})^{\top}$.

We use again the map \eqref{eq:map}. For any $z \in \Tilde{Y}$, we have $\hat{T}(z)=T(D^{-1}_f F(z) +D^{-1}_f H(z))\geq T(D^{-1}_f F(z)) \geq T(D^{-1}_f F(c)) \geq c$. We use again the Brouwer Fixed-Point Theorem \cite{shapiro2016fixed} and confirm that there exists an endemic equilibrium $z^*\gg \mathbf{0}$ if $D^{-1}_f F(c)\geq c$.
\end{proof}

\begin{remark}
The condition $D^{-1}_f F(c)\geq c$ is not strict. We consider $c$ now very close to the all-zero vector. We expand the function $F$ around $0$ and get $D_f^{-1} (F(0)+ J_F(0)c) \geq c$. Since $F(0)=0$, we finally get $ J_F(0) c\geq D_f c$, which implies the influence of infection rates around zero is stronger than the influence of healing rates.
\end{remark}

Next, we can further show that the general system \eqref{eq::sys_siws_hyp_z2} is also an irreducible monotone system.

\begin{theorem} [Irreducible monotone system]
If Assumption \ref{ass:xini}-\ref{ass:irre3} hold,\\ the model \eqref{eq::sys_siws_hyp_z2} is an irreducible monotone system in $\mathbf{{D}}$. Furthermore, if the model \eqref{eq::sys_siws_hyp_z2} has a finite number of equilibria in the closure of $\mathbf{\Bar{D}}$ for a generic choice of parameters, then the model converges to an equilibrium for almost all initial conditions. That is, the set of initial conditions, such that the model does not converge to an equilibrium, is a set of Lebesgue measure zero.
\end{theorem}

\begin{proof}
Firstly, we calculate the Jacobian
\begin{align*}
    J=-D_f+(I-Z)(J_F(z)+J_H(z))-Z\Dg(F(z)+H(z));
 \end{align*}
Note that $J_F(z)+J_H(z)> \mathbf{0}$. Furthermore, $D_f$ is diagonal. Thus, the Jacobian is an irreducible Metzler matrix at an arbitrary point in $\mathbf{{D}}$ along the trajectories of \eqref{eq::sys_siws_hyp_z2}. This confirms that the system \eqref{eq::sys_siws_hyp_z2} is an irreducible monotone system by Lemma 2.2 of \cite{ye2022convergence}.

By assumption, the equilibrium set is finite. Thus, by Lemma 2.3 of \cite{ye2022convergence}, the model converges to an equilibrium with some domain of attraction. The set of the initial conditions, such that the model does not converge to an equilibrium, is a set of measure zero.
\end{proof}

\begin{remark}
    If Assumption \ref{ass:xini}-\ref{ass:ho} hold, and furthermore the pairwise interaction and the higher order interaction is a polynomial function, the model \eqref{eq::sys_siws_hyp_z2} under this setting has a finite number of equilibria in the closure of $\mathbf{\Bar{D}}$ for a generic choice of parameter, then the model converges to an equilibrium for almost all initial conditions. That is, the set of the initial conditions, such that the model does not converge to an equilibrium, is a set of Lebesgue measure zero.
\end{remark}

Furthermore, if $\rho(D^{-1}_f B_f)>1$, the system with pairwise interaction is similar to the conventional system \eqref{eq::sys_siws_hyp_z}  and without the higher-order interaction, the system is reduced to the conventional continuous-time SIWS model on a graph \cite{pare2022multi}. Now, the system on a graph reads as
\begin{equation} \label{eq::siws_graph}
    \dot{z}=-D_f z+(I-{Z})B_f z.
\end{equation}

\begin{prop}
Consider the system \eqref{eq::siws_graph} with $D_f$ positive diagonal and $B_f$ irreducible non-negative. If $\rho(D^{-1}_f B_f)>1$, then the system has a unique endemic equilibrium (strictly positive equilibrium) and it is hyperbolic and globally stable in $\Bar{\mathbf{D}}\setminus \{\mathbf{0}\}$.
\end{prop}

\begin{proof}
The existence, uniqueness, and global stability of the endemic equilibrium are reported in theorem 2 in  \cite{pare2022multi}. So, we only show that the endemic equilibrium is hyperbolic, which means that the Jacobian at the endemic equilibrium has no eigenvalues with a zero real part.

So we can see that the Jacobian is $J=-D_f+B_f-{Z^*}B_f-B_f{Z^*}$. We notice that $Jz^*=-D_fz^*+B_fz^*-{Z^*}B_fz^*-B_f{Z^*}z^*=-B_f{Z^*}z^*<-dz^*$ with some positive $d$. Thus, $J$ is Hurwitz and contains only eigenvalues with a negative real part. Thus, the endemic equilibrium is hyperbolic.
\end{proof}

Let us now consider the system \eqref{eq::siws_graph} with $D_f$ positive diagonal and $B_f$ irreducible non-negative with some small higher-order term as $\epsilon H(z)$ satisfying Assumption \ref{ass:f0}-\ref{ass:ho}. That is,

\begin{equation} \label{eq::siws_per}
    \dot{z}=-D_f z+(I-{Z})(B_f z +\epsilon H(z)).
\end{equation}

\begin{theorem}
Consider the perturbed system \eqref{eq::siws_per}. If Assumptions \ref{ass:xini}-\ref{ass:irre3} hold, and $\rho(D^{-1}_f B_f)>1$, then there is at least one endemic equilibrium, which is hyperbolic and locally stable, for $0<\epsilon\ll 1$ sufficiently small.
\end{theorem}

\begin{proof}
We know from \cite{pare2022multi} that if $\rho(D^{-1}_f B_f)>1$, the unperturbed system has a unique endemic equilibrium and it is hyperbolic. We can rewrite the unperturbed system as $g(z)$ and the perturbed system as $G(z,\epsilon)$. Let $z^*$ be the unique endemic equilibrium of the unperturbed system. By definition, $G(z^*,0)=0$ and $\frac{\partial G}{\partial x}(z^*,0)=\frac{\partial f}{\partial x}(z^*)$. Since $z^*$ is hyperbolic,  $\frac{\partial G}{\partial x}(z^*,0)=\frac{\partial f}{\partial x}(z^*)$ has a nonvanishing determinant. By the implicit function theorem, there is a unique equilibrium in the neighborhood of $z^*$ for sufficiently small $\epsilon$. This equilibrium is also hyperbolic by continuous dependence of the eigenvalues of $\frac{\partial G}{\partial x}$ on $\epsilon$. Thus, the local stability of the equilibrium also persists.
\end{proof}

 The model \eqref{eq::sys_siws_hyp_z2} can be easily extended into a competing bi-virus SIWS model, which reads as
\begin{align}
    \dot{z}^1_i&=\dot{x}^1_i = -\delta^1_i x^1_i 
    + (1-x^1_i-x^2_i)  \left( f^1_i(z^1) 
 +  h^1_i(z^1) \right), &\quad 1\leq i \leq n,
 \label{eq::sys_siws_hyp_xi1}\\
 \dot{z}^1_i &=\dot{w}^1_{j} = -\delta^{w1}_{j} w^1_j+ f^1_i(x^1) 
 +  h^1_i(x^1), &\quad n+1\leq i=j+n \leq n+m ,\label{eq::sys_siws_hyp_wi1}\\
 \dot{z}^2_i&=\dot{x}^2_i = -\delta^2_i x^2_i 
    + (1-x^1_i-x^2_i)  \left( f^2_i(z^2) 
 +  h^2_i(z^2) \right), &\quad 1\leq i \leq n
 \label{eq::sys_siws_hyp_xi2}\\
 \dot{z}^2_i &=\dot{w}^1_{j} = -\delta^{w2}_{j} w^2_j+ f^2_i(x^2) 
 +  h^2_i(x^2), &\quad n+1\leq i=j+n \leq n+m,\label{eq::sys_siws_hyp_wi2}
\end{align}
where all notations remain the same as in \eqref{eq:notation1}, \eqref{eq:notation2} and the superscript $1,2$ denotes the first or the second virus.

For $\nu=1,2$, we can rewrite \eqref{eq::sys_siws_hyp_xi1}-\eqref{eq::sys_siws_hyp_wi2} as
\begin{equation}\label{eq::sys_siws_dt_zbi}
\begin{split}
    \dot{z}^\nu&=-D_f^\nu z^\nu +(I-Z^1-Z^2)(F^\nu(z^\nu)+H^\nu(z^\nu)).
\end{split}
\end{equation}

Throughout the rest of this paper, we assume that if Assumption \ref{ass:f0}-\ref{ass:ho} hold for the system \eqref{eq::sys_siws_dt_zbi}, then Assumptions \ref{ass:f0}-\ref{ass:ho} hold for each $F^\nu(z^\nu), H^i(z^\nu), A^\nu, \nu=1,2$.

For the bi-virus system, we have a new assumption on the initial condition.
\begin{assumption}\label{ass:xinibi}
The initial condition of infection level satisfies $x^\nu_i(0)\in [0,1]$ for all $i$ and $\nu$. Furthermore, $x^1_i(0)+x^2_i(0) \leq 1$ holds for all $i=1,\cdots,n$.
\end{assumption}

\begin{lemma}[General properties]\label{lem:2}
Consider the bi-virus model \eqref{eq::sys_siws_dt_zbi} and let $w^\nu_{j\max}=\frac{f^\nu_j(1)+h^\nu_j(1)}{\delta_j^{w\nu}}$, 
\begin{equation}
    \begin{split}
        \mathbf{D} &=\left\{z=(x^1,w^1,x^2,w^2)^{\top}=(x^1_1,\cdots,x^1_n,w^1_1,\cdots,w^1_m,x^2_1,\cdots,x^2_n,w^2_1,\cdots,w^2_m)^\top| \right.\\
        &\qquad\left. x^\nu_i \in (0,1), w^\nu_j \in (0,w^\nu_{j\max}), \nu=1,2 \right\},
    \end{split}
\end{equation}

and $\mathbf{\Bar{D}}$ be its closure. If Assumptions \ref{ass:f0}-\ref{ass:xinibi} hold, then
\begin{itemize}
    \item [i)] the set $\mathbf{\Bar{D}}$ is positively invariant;
    \item [ii)] the set $\mathbf{D}$ is also positively invariant;
    \item [iii)] the origin is always an equilibrium of the model;
    \item [iv)] the dynamics of each virus \eqref{eq::sys_siws_dt_zbi} are upper bounded by its single virus counterpart \eqref{eq::sys_siws_hyp_z2} with the same system parameters.
\end{itemize}
\end{lemma}

\begin{proof}
If $x^\nu_i=0$, we have $\dot{x}^\nu_i= (1-x^{3-\nu}_i) (f^\nu_i(z) +  h^\nu_i(z))  \geq 0$. If $w^\nu_j=0$, $\dot{w}^\nu_j=f^\nu_{j+n}(x) +  h^\nu_{j+n}(x)\geq 0$. If $x^\nu=\mathbf{0}$ and $w^\nu=\mathbf{0}$, we have $\dot{x}= \mathbf{0}$ and $\dot{w}= \mathbf{0}$, which shows that the origin is always an equilibrium.

If $x^\nu_i=1$, we get $\dot{x}^\nu_i=-\delta_i<0$. If $w_j^\nu=w^\nu_{j\max}$, we have $\dot{w}^\nu_j=-f^\nu_{n+j}(\mathbf{1})-h^\nu_{n+j}(\mathbf{1})+ f^\nu_{n+j}(x^\nu)+h^\nu_{n+j}(x^k)\leq 0$, since the pairwise interaction is increasing in $x$ and the higher-order interaction function is non-decreasing in $x$. If $w^\nu_j=w^\nu_{j\max}$ and $x^\nu\neq \mathbf{1}$, then $\dot{w}^\nu_j< 0$.

If $x^1_i+x^2_i=1$, then $x^\nu_i=-\delta_i x^\nu_i<0$. All these arguments show that statements i) and ii) are true.

The statement iv) is straightforward. It suffices to compare the dynamics of each virus \eqref{eq::sys_siws_dt_zbi} with its single virus counterpart with the same system parameters \eqref{eq::sys_siws_hyp_z2}.
\end{proof}

\begin{theorem}[Healthy-state behaviour] \label{thm:healthybi2}
Consider the model \eqref{eq::sys_siws_dt_zbi}, if Assumptions \ref{ass:f0}-\ref{ass:xinibi} hold, we have the following statements.
The healthy state (zero equilibrium)
\begin{itemize}
    \item [1)] is locally stable if $\rho((D^{\nu}_f)^{-1} J_{F^\nu}(0))<1$ for $\nu=1,2$;
    \item [2)] is globally exponentially stable if $(D^{\nu}_f)^{-1} F^\nu(z^{\nu})+(D^{\nu}_f)^{-1} H^\nu(z^{\nu})\leq dz^{\nu}$, $\nu=1,2$ for some  $0< d < 1$;
    \item [3)] is unstable if $\rho((D^{\nu}_f)^{-1} J_{F^\nu}(0))>1$ for either $\nu=1$ or $\nu=2$.
\end{itemize}
\end{theorem}

\begin{proof}
It is already clear that the healthy state (zero equilibrium) always exists. The Jacobian matrix of the model at the origin reads as
 \begin{equation*}
     J(0,0)=\left[ 
        \begin{matrix}
	-D^1_f+J_{F^1}(0)& \mathbf{0}  \\
	\mathbf{0} & -D^2_f+J_{F^2}(0)   \\
	\end{matrix}
        \right].
 \end{equation*}
 
Thus, the healthy state is locally stable if $\rho((D^{\nu}_f)^{-1} J_{F^\nu}(0))<1$ for $\nu=1,2$ and unstable if $\rho((D^{\nu}_f)^{-1} J_{F^\nu}(0))>1$ for either $\nu=1$ or $\nu=2$.

The global exponential stability when  $(D^{\nu}_f)^{-1} F^\nu(z^{\nu})+(D^{\nu}_f)^{-1} H^\nu(z^{\nu})\leq dz^{\nu}$ for some $0< d < 1$ and $\nu=1,2$ is the direct consequence of statement iv) in Lemma \ref{lem:2} and Theorem \ref{thm:healthy}.
\end{proof}

\begin{theorem} [Dominant endemic equilibrium]
Consider the bi-virus model \\ \eqref{eq::sys_siws_dt_zbi}. If Assumptions \ref{ass:f0}-\ref{ass:xinibi} hold, we have the following statements. If $z^*=(x^*,w^*)^\top$ is the endemic equilibrium of the single-virus counterpart \eqref{eq::sys_siws_hyp_z} of the first virus, then $(z^*,0,0)^{\top}$ is the dominant endemic equilibrium of the system \eqref{eq::sys_siws_dt_zbi}. The dominant endemic equilibrium is locally stable when $z^*$ is locally stable and $\rho(-D^2_{f}+(I-Z^*)J_{F^2}(0))<0$ with respect to  the single-virus system and is unstable when $z^*$ is unstable or $\rho(-D^2_{f}+(I-Z^*)J_{F^2}(0))>0$, where $Z^*=\left[ 
        \begin{matrix}
	x^* & 0  \\
	0 & 0  \\
	\end{matrix}
        \right]$.
\end{theorem}

\begin{proof}
If we set $z^2=\mathbf{0}$, the system is reduced into the single virus counterpart of the first virus. Thus, if $z^*$ is the endemic equilibrium of the single-virus counterpart \eqref{eq::sys_siws_hyp_z} of the first virus, then $(z^*,0,0)^{\top}$ is the dominant endemic equilibrium of the system \eqref{eq::sys_siws_dt_zbi}.

The Jacobian matrix of the model at the dominant endemic equilibrium is 
 \begin{equation*}
     J(z^*,\mathbf{0})=\left[ 
        \begin{matrix}
	\hat{J}(z^*) & \mathbf{Res}  \\
	\mathbf{0} & -D^2_{f}+(I-Z^*)J_{F^2}(0)   \\
	\end{matrix}
        \right],
 \end{equation*}
 where $\hat{J}(z^*)$ is the Jacobian of the single-virus counterpart at the endemic equilibrium $z^*$ and $\mathbf{Res}$ is the matrix of some residual which we do not need to deal with.
 
 By observing the Jacobian and noticing that $\hat{J}(z^*)$ and $-D^2_{f}+(I-Z^*)J_{F^2}(0)$ is irreducible nonnegative, we can directly obtain the results.
\end{proof}

\begin{theorem} [Coexisting equilibrium]
Consider the model \eqref{eq::sys_siws_dt_zbi}. If Assumptions \ref{ass:f0}-\ref{ass:xinibi} hold,  we have the following statements. 
If there exist $(D^{\nu}_f)^{-1} F^\nu(c)\geq c$ and $c\neq 0$, $c\ll u^i$, $F^\nu(z^\nu)\geq B^\nu_f z^\nu$ for both viruses $\nu=1,2$, and $\rho(-D^1_{f}+(I-Z^{2*})B^1_f)>0$ as well as $\rho(-D^2_{f}+(I-Z^{1*})B^2_f)>0$, where $z^{1*}$ and $z^{2*}$ are the endemic equilibrium of the single-virus counterpart of the first and second virus respectively and always exist, then there exists at least one coexisting equilibrium $(\hat{z}^1,\hat{z}^2)\gg \mathbf{0}$ such that $\hat{x}^1+\hat{x}^2\leq \mathbf{1}$.
\end{theorem}

\begin{proof}
First of all, we need to confirm that $z^{1*}$ and $z^{2*}$, the endemic equilibrium of the single-virus counterpart of the first and second virus, always exist under the conditions. Since there must exist $(D^{\nu}_f)^{-1} F^\nu(c)\geq c$ for both viruses $\nu=1,2$, according to the Theorem \ref{thm:end2ab}, $z^{1*}$ and $z^{2*}$ must always exist.

Now, define the map $M$ as follows. Component-wise, it is given by 
\begin{equation}
\begin{aligned}
M_i^1\left(z^1, z^2\right) &=\frac{\left(1-z_i^2\right)\left(\left(D_f^1\right)^{-1}F^1(z^1)+\left(D_f^1\right)^{-1} H^1(z^1)\right)_i}{1+\left(\left(D_f^1\right)^{-1}F^1(z^1) +\left(D_f^1\right)^{-1} H^1(z^1) \right)_i} \\
M_i^2\left(z^1, z^2\right) &=\frac{\left(1-z_i^1\right)\left(\left(D_f^2\right)^{-1}F^2(z^2)+\left(D_f^2\right)^{-1} H^2(z^2)\right)_i}{1+\left(\left(D_f^2\right)^{-1}F^2(z^2)+\left(D_f^2\right)^{-1} H^2(z^2)\right)_i}
\end{aligned}
\end{equation}
for $i\leq n$.
Furthermore, for $i\geq n+1$,
\begin{equation}
\begin{aligned}
M_i^1\left(z^1, z^2\right) &=(D_f^1)^{-1} (F^1(z^1)+H^1(z^1)) \\
M_i^2\left(z^1, z^2\right) &=(D_f^2)^{-1} (F^2(z^2)+H^2(z^2)).
\end{aligned}
\end{equation}

Similarly, by using the property of the monotonicity of $M$, we get $M(z^1,z^2)\leq M(z^{1*},z^{2*})$ if $(z^1,z^2)\leq(z^{1*},z^{2*})$.

Now, we construct another map $\hat{M}$.  Component-wise, it is given by, for $i\leq n$,
\begin{equation}
\begin{aligned}
\hat{M}_i^1\left(z^1, z^2\right) &=\frac{\left(1-z_i^2\right)\left(\left(D_f^1\right)^{-1} B_f^1 z^1\right)_i}{1+\left(\left(D_f^1\right)^{-1} B_f^1 z^1  \right)_i}\leq M_i^1\left(z^1, z^2\right) \\
\hat{M}_i^2\left(z^1, z^2\right) &=\frac{\left(1-z_i^1\right)\left(\left(D_f^2\right)^{-1} B_f^2 z^2\right)_i}{1+\left(\left(D_f^2\right)^{-1} B_f^2 z^2\right)_i}\leq M_i^2\left(z^1, z^2\right);
\end{aligned}
\end{equation}
these inequalities are true because the function $\frac{s}{1+s}$ is increasing as $s$ increases, and furthermore $F^\nu(z^\nu)+H^\nu(z^\nu)\geq F^\nu(z^\nu) \geq B_f^\nu z^\nu$. 

Furthermore, for $i\geq n+1$,
\begin{equation}
\begin{aligned}
\hat{M}_i^1\left(z^1, z^2\right) &=(D_f^1)^{-1} (B_f^1 z^1)\leq M_i^1\left(z^1, z^2\right) \\
\hat{M}_i^2\left(z^1, z^2\right) &=(D_f^2)^{-1} (B_f^2 z^2)\leq M_i^2\left(z^1, z^2\right).
\end{aligned}
\end{equation}

The rest proof remains the same as for Theorem \ref{thm:coexist}.
\end{proof}

\begin{remark}
The condition $F^\nu(z^\nu)\geq B^\nu_f z^\nu$ shows that the pairwise interaction is lower bounded by a linear function. Thus, the rate of the increase of the pairwise interaction must exceed the slope of the linear function. Alternatively, the general higher-order model can be seen as lower-bounded by the conventional SIWS model on a graph, since the linear function denotes the infection process of the conventional SIWS model on a graph.
\end{remark}

\begin{theorem} [Irreducible monotone system]
If Assumptions \ref{ass:f0}-\ref{ass:xinibi} hold, the model \eqref{eq::sys_siws_dt_zbi} is an $(\mathbf{0}_{n+m},\mathbf{1}_{n+m})-$type irreducible monotone system in $\mathbf{{D}}$. Furthermore, if the model \eqref{eq::sys_siws_dt_zbi} has a finite number of equilibria in the closure of $\mathbf{\Bar{D}}$ for a generic choice of parameter, then the model converges to an equilibrium for almost all initial conditions. That is to say, the set of initial conditions, such that the model does not converge to an equilibrium, is a set of Lebesgue measure zero.
\end{theorem}

\begin{proof}
The proof is similar to the Theorem \ref{thm:irsbi}. Notice that the Jacobian can be permutated into an irreducible Metzler matrix via the permutation matrix $P=\Dg((1,\cdots,1,-1,\cdots,-1)^\top)$. The rest of the proof remains the same. 
\end{proof}

\begin{remark}
If the pairwise and the higher-order interactions are both polynomial functions, then the model \eqref{eq::sys_siws_dt_zbi} has a finite number of equilibria in the closure of $\mathbf{\Bar{D}}$ for a generic choice of parameter. As a consequence, the model converges to an equilibrium for almost all initial conditions.
\end{remark}

\section{Numerical examples}\label{sec:sim}
In this section, we carry out numerical simulations to illustrate our analytical results.

As for the simulation setup, we consider a hypergraph with 5 population nodes and 2 resource nodes. The recovery rate and the decay rate are randomly picked up from $[0,1]$. Furthermore, the infection rate and the contamination rate are randomly picked up from $[0,0.2]$. By this setting, the Assumptions \ref{ass:xini1}-\ref{ass:irre} hold already. We observe that the parameters of this setting usually reasonably provide a reproduction number. Otherwise, for example, if we allow the maximum infection rate and the maximum contamination rate to exceed $0.2$, then it is very likely to observe a super big reproduction number. For simplicity, each figure of the simulation represents the mean value of the infection (or contamination) level of all population (or resource) nodes. In each figure, the short name \emph{IC} denotes initial condition and \emph{V} means virus.

\subsection{Conventional higher-order single- and bi-virus system}
Firstly, we make a simulation on the  single-virus system \eqref{eq::sys_siws_hyp_z} and the bi-virus system \eqref{eq::sys_siws_dt_zbi1}. The simulation results of figure \ref{fig:sig-healthy}-\ref{fig:bi-end3} are in line with the analytical results of this paper. Especially, we see that bi-stability occurs for both single-virus and bi-virus systems. Although we confirm the existence of coexisting equilibrium under some conditions, the simulation results suggest that it is unstable.
\begin{figure}
    \centering
    \includegraphics{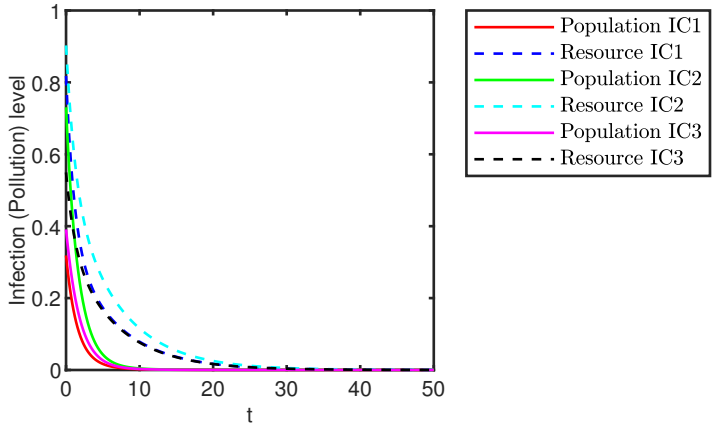}
    \caption{The simulation of the single-virus case when the reproduction number equals $0.0644$. From 3 different random initial conditions, the model always converges to zero.}
    \label{fig:sig-healthy}
\end{figure}

\begin{figure}
    \centering
    \includegraphics{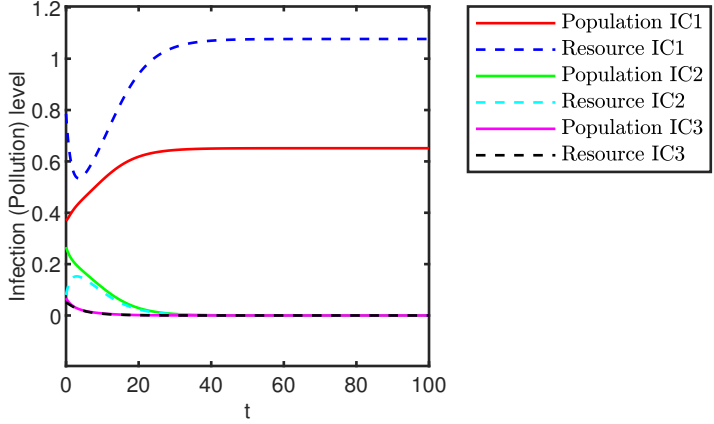}
    \caption{The simulation of the single-virus case when the reproduction number equals $0.4578$. From 3 different random initial conditions, the model either converges to the origin or to an endemic equilibrium.}
    \label{fig:bistability}
\end{figure}

\begin{figure}
    \centering
    \includegraphics{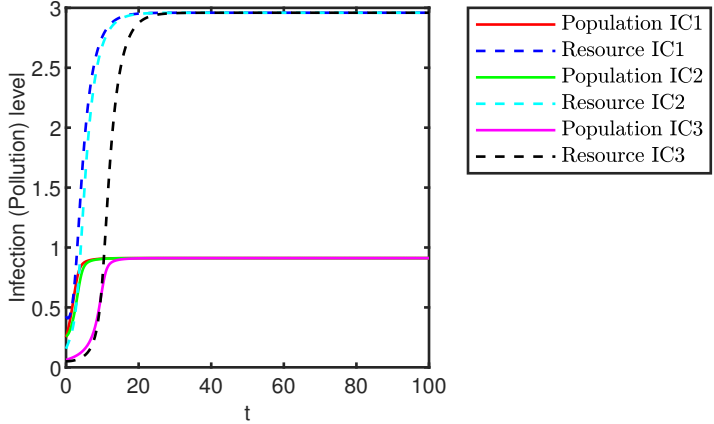}
    \caption{The simulation of the single-virus case when the reproduction number equals $1.7052$. From 3 different random initial conditions, the model always converges to the same endemic equilibrium.}
    \label{fig:sig-end}
\end{figure}

\begin{figure}
    \centering
    \includegraphics{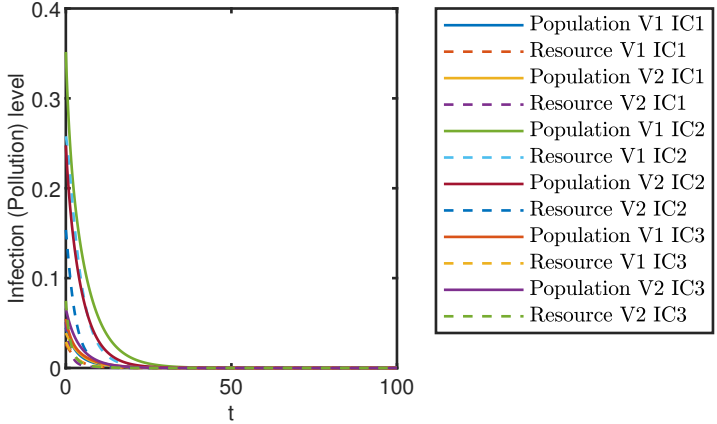}
    \caption{The simulation of the bi-virus case when the reproduction numbers equal $0.3632$ and $0.4563$ respectively. From 3 different random initial conditions, the model always converges to the healthy state.}
    \label{fig:bi-healthy}
\end{figure}

\begin{figure}
    \centering
    \includegraphics{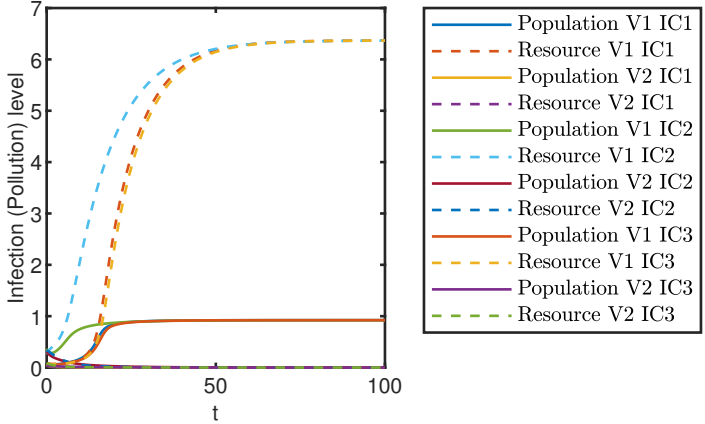}
    \caption{The simulation of the bi-virus case when the reproduction numbers equal $2.3867$ and $0.5647$ respectively. From 3 different random initial conditions, the model always converges to the same dominant endemic equilibrium.}
    \label{fig:bi-end1}
\end{figure}

\begin{figure}
    \centering
    \includegraphics{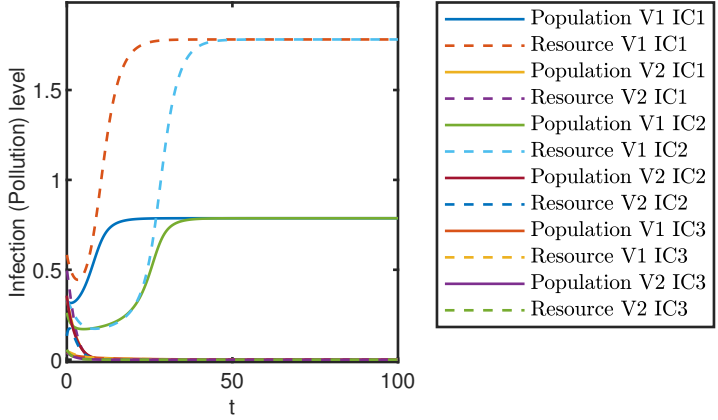}
    \caption{The simulation of the bi-virus case when the reproduction numbers equal $0.6640$ and $0.4685$ respectively. From 3 different random initial conditions, the model either converges to a dominant endemic equilibrium or to the healthy state.}
    \label{fig:bi-end2}
\end{figure}

\begin{figure}
    \centering
    \includegraphics{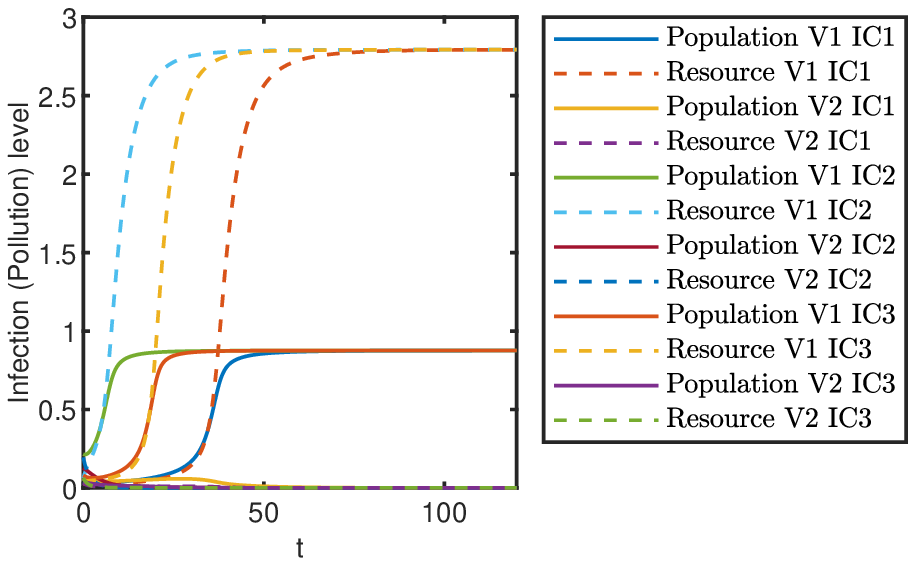}
    \caption{The simulation of the bi-virus case when the reproduction numbers are the same and equal $1.4575$. From 3 different random initial conditions, the model always converges to a dominant endemic equilibrium.}
    \label{fig:bi-end3}
\end{figure}

\subsection{Bi-virus system with HOIs up to 4-body interactions}
Here, we simulate the case when we consider the HOIs up to 4-body interactions. That is to consider bi-virus system \eqref{eq::sys_siws_dt_zbi} with pairwise interactions $F^l(z)=B^l_f z$ and HOIs $h^l_i(z)=\sum_{j,k\in \mathbf{N}_3^{\text{i}}} \beta^l_{i3} A^l_{ijk} z^l_jz^l_k + \sum_{j,k,l\in \mathbf{N}_4^{\text{i}}} \beta_{i4} A_{ijkl} z_jz_kz_l, l=1,2$. From figure \ref{fig:bistability-sig-cub} and \ref{fig:bivirus-cub}, we see the system still processes bistability. These numerical results support our analytical results of a generalized SIWS model. Moreover, the numerical examples further suggest that any coexisting equilibrium is unstable since we don't observe any solution converging to any coexisting equilibrium by calibrating system parameters and initial conditions.

\begin{figure}
    \centering
    \includegraphics{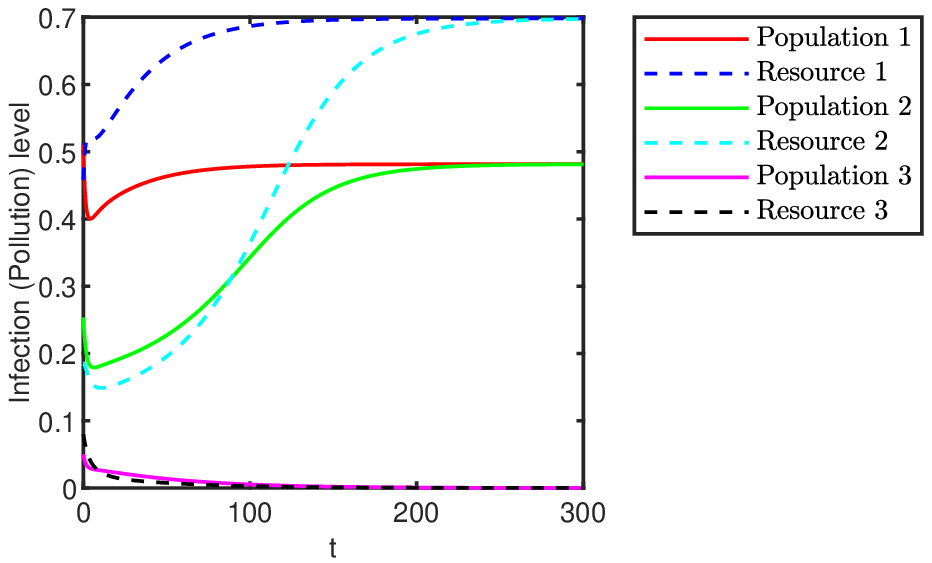}
    \caption{The simulation of the single-virus case. From 3 different random initial conditions, the model either converges to an endemic equilibrium or to the healthy state.}
    \label{fig:bistability-sig-cub}
\end{figure}

\begin{figure}
    \centering
    \includegraphics{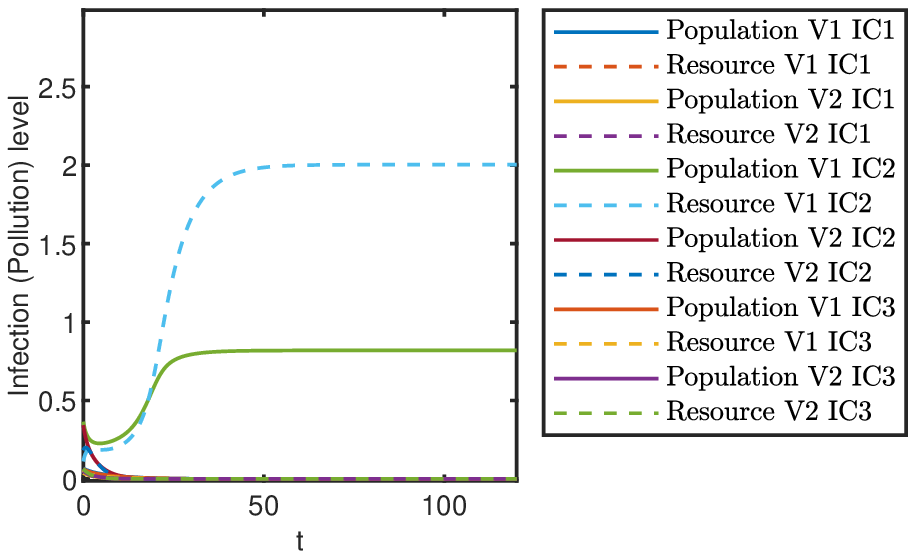}
    \caption{The simulation of the bi-virus case. From 3 different random initial conditions, the model converges either to a dominant endemic equilibrium or to the healthy state.}
    \label{fig:bivirus-cub}
\end{figure}

\subsection{Bi-virus system with non-linear pairwise interactions and polynomial HOIs}
In this subsection, we simulate the case when we consider pairwise interactions are non-linear and the HOIs are polynomial functions. That is to consider bi-virus system \eqref{eq::sys_siws_dt_zbi} with pairwise interactions $f^l_i(z)=\sum_{j,k\in \mathbf{N}_2^{\text{i}}} \beta^l_{i2} A^l_{ij} \log(1+0.2z^l_j)$ and HOIs $h^l_i(z)=\sum_{j,k\in \mathbf{N}_3^{\text{i}}} \beta^l_{i3} A_{ijk} z^l_jz^l_k, l=1,2$. These pairwise interactions are the same as the form suggested by \cite{doshi2022convergence} and the HOIs are the same form as the conventional single-virus system \eqref{eq::sys_siws_hyp_z} and the bi-virus system \eqref{eq::sys_siws_dt_zbi1}.  From figure \ref{fig:bistability-sig-1log} and \ref{fig:biviruslog1}, we see the system still processes bistability. These numerical results again support our analytical results of a generalized SIWS model. Furthermore, by calibrating system parameters and initial conditions, we don't observe that any solution converges to any coexisting equilibrium. Coexisting equilibrium seems always unstable.

\begin{figure}
    \centering
    \includegraphics{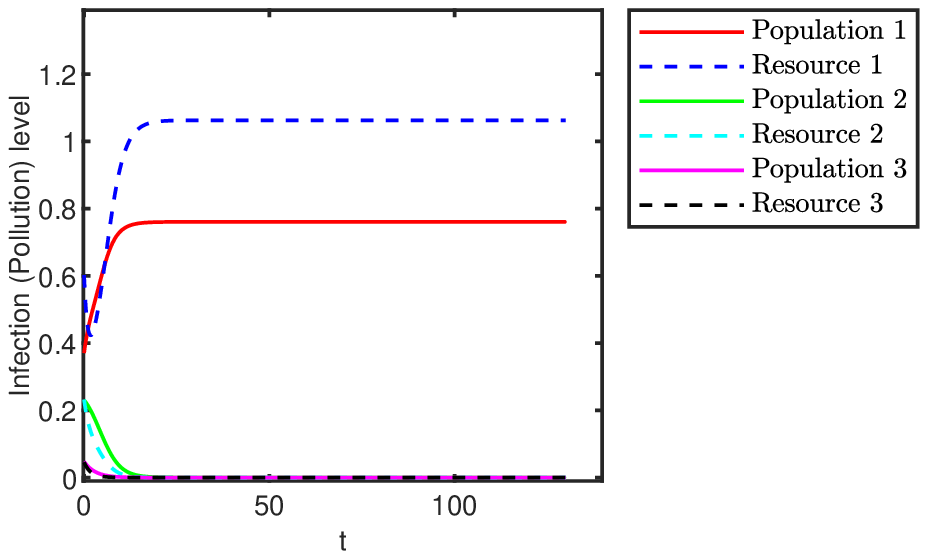}
    \caption{The simulation of the single-virus case. From 3 different random initial conditions, the model either converges to an endemic equilibrium or to the healthy state.}
    \label{fig:bistability-sig-1log}
\end{figure}

\begin{figure}
    \centering
    \includegraphics{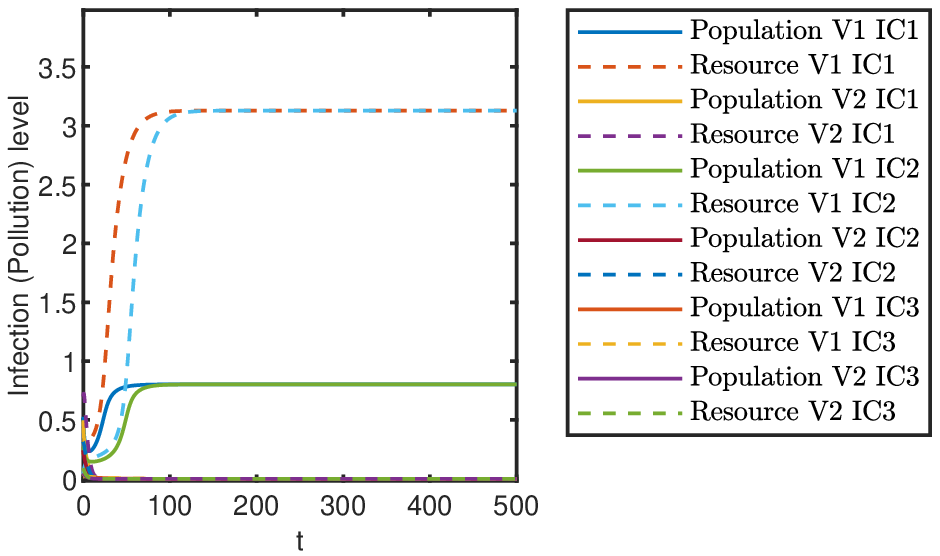}
    \caption{The simulation of the bi-virus case. From 3 different random initial conditions, the model converges either to a dominant endemic equilibrium or to the healthy state.}
    \label{fig:biviruslog1}
\end{figure}

\subsection{Bi-virus system with non-linear pairwise interactions and non-polynomial HOIs}
In this subsection, we simulate the case when we consider pairwise interactions are non-linear and the HOIs are in a similar form. That is to consider bi-virus system \eqref{eq::sys_siws_dt_zbi} with pairwise interactions $f^l_i(z)=\sum_{j,k\in \mathbf{N}_2^{\text{i}}} \beta^l_{i2} A^l_{ij} \log(1+0.2z^l_j)$ and HOIs $h^l_i(z)=\sum_{j,k\in \mathbf{N}_3^{\text{i}}} \beta^l_{i3} A^l_{ijk} z^l_j\log(1+0.2z^l_k), l=1,2$. From figure \ref{fig:multistability-sig-2log} and \ref{fig:biviruslog2}, we see the system processes bistability. These numerical results are still in line with our analytical results of a generalized SIWS model. Furthermore, by calibrating system parameters and initial conditions, we don't observe that any solution converges to any coexisting equilibrium. Coexisting equilibrium seems always unstable.

\begin{figure}
    \centering
    \includegraphics{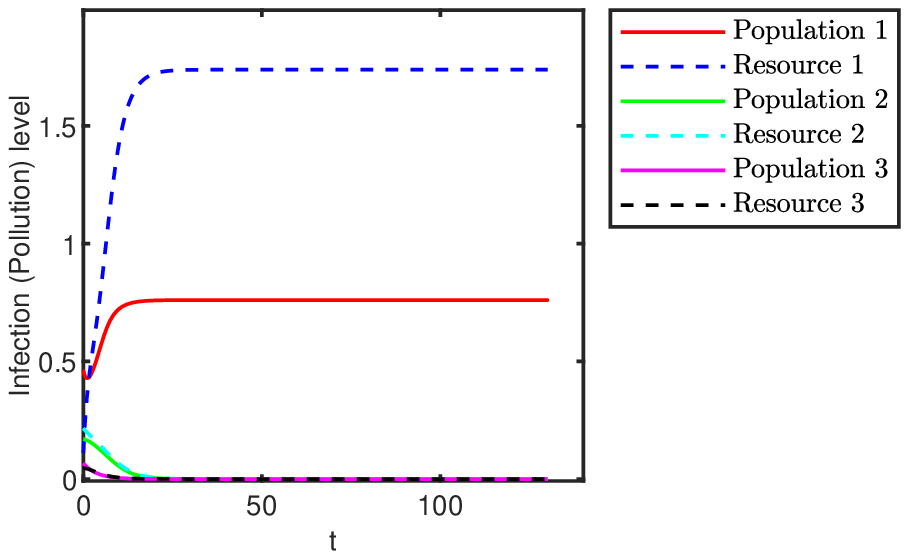}
    \caption{The simulation of the single-virus case. From 3 different random initial conditions, the model either converges to an endemic equilibrium or to the healthy state.}
    \label{fig:multistability-sig-2log}
\end{figure}

\begin{figure}
    \centering
    \includegraphics[height=5cm]{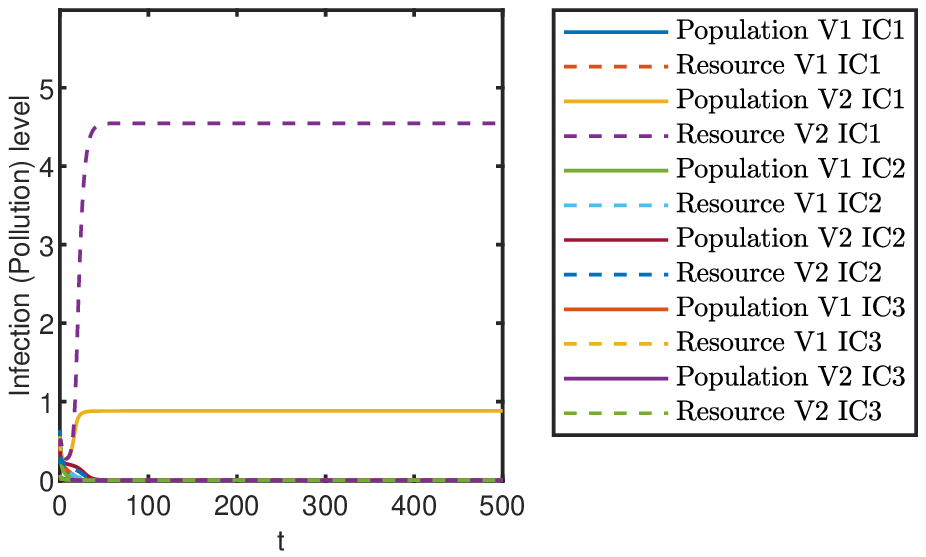}
    \caption{The simulation of the bi-virus case. From 3 different random initial conditions, the model converges either to a dominant endemic equilibrium or to the healthy state.}
    \label{fig:biviruslog2}
\end{figure}

\section{Conclusion and future works}
In this paper, we study a general SIS-type diffusion process with both indirect and direct pathways on a directed hypergraph. Based on a polynomial interaction function, although the interaction among agents is firstly based on a simplicial complex structure, we show that such a model is eventually equivalent to a model on a general hypergraph if we adopt a spreading mechanism of different orders. From a formal level, this has allowed us to study a spreading process on a rather general hypernetwork setting. We further give a full analysis of our proposed model. The system analysis mainly focuses on the healthy-state and the endemic behaviors. Importantly, induced by the higher order interaction, we find out that bistability may even occur when the reproduction number is smaller than $1$. When the reproduction number is greater than $1$, if the higher-order terms are sufficiently small, the unique endemic equilibrium is globally stable. Then, we continue to propose a competing bi-virus  model on a directed hypergraph and the system analysis is analogously carried out. The bistability may also occur in the bi-virus system which is a direct consequence of the bistability of a single-virus system. Particularly, we give a condition for the system having a coexisting equilibrium. However, further numerical studies show that coexisting equilibrium seems always unstable. To further generalize the proposed framework, we consider an abstract interaction function, which in real-life processes may account for interactions that do not follow a polynomial rule. The corresponding single-virus and bi-virus models are proposed and studied. The global system behavior is generally similar to the counterpart with a polynomial interaction function. In addition, we show that all proposed systems in this paper belong to irreducible monotone systems. Under the further assumption of a finite number of equilibria, solutions of these systems always converge to an equilibrium. The theoretical results are finally supported by some numerical examples.

Several future research directions stem from the work presented here; these include, but are not limited to, proving the instability of a coexisting equilibrium, providing further results on the endemic behavior (stability) of the general system, and designing control strategies to regulate the spreading process (e.g. adding or eliminating a node or a hyperedge from a hypergraph). For the last item, we may need to investigate tensor spectral theory \cite{chang2013survey,qi2017tensor}, since a hypergraph can be abstracted as tensors, and manipulation of such a hypergraph may change the eigenvalues of the tensor, which may further influence the system's behavior.




\bibliographystyle{siamplain}
\bibliography{references}

\begin{thebibliography}{10}

\bibitem{achterberg2022comparing}
{\sc M.~A. Achterberg, B.~Prasse, L.~Ma, S.~Trajanovski, M.~Kitsak, and
  P.~Van~Mieghem}, {\em Comparing the accuracy of several network-based
  covid-19 prediction algorithms}, International journal of forecasting, 38
  (2022), pp.~489--504.

\bibitem{bick2021higher}
{\sc C.~Bick, E.~Gross, H.~Harrington, and M.~Schaub}, {\em What are higher
  order networks?}, SIAM Review,  (2022).

\bibitem{FB-LNS}
{\sc F.~Bullo}, {\em Lectures on Network Systems}, Kindle Direct Publishing,
  {1.6}~ed., 2022, \url{http://motion.me.ucsb.edu/book-lns}.

\bibitem{carli2020model}
{\sc R.~Carli, G.~Cavone, N.~Epicoco, P.~Scarabaggio, and M.~Dotoli}, {\em
  Model predictive control to mitigate the covid-19 outbreak in a multi-region
  scenario}, Annual Reviews in Control, 50 (2020), pp.~373--393.

\bibitem{chang2013survey}
{\sc K.~Chang, L.~Qi, and T.~Zhang}, {\em A survey on the spectral theory of
  nonnegative tensors}, Numerical Linear Algebra with Applications, 20 (2013),
  pp.~891--912.

\bibitem{chaves2007loss}
{\sc S.~S. Chaves, P.~Gargiullo, J.~X. Zhang, R.~Civen, D.~Guris, L.~Mascola,
  and J.~F. Seward}, {\em Loss of vaccine-induced immunity to varicella over
  time}, New England Journal of Medicine, 356 (2007), pp.~1121--1129.

\bibitem{cisneros2021multi}
{\sc P.~Cisneros-Velarde and F.~Bullo}, {\em Multi-group sis epidemics with
  simplicial and higher-order interactions}, IEEE Transactions on Control of
  Network Systems,  (2021).

\bibitem{cui2022discrete}
{\sc S.~Cui, F.~Liu, H.~Jard{\'o}n-Kojakhmetov, and M.~Cao}, {\em Discrete-time
  layered-network epidemics model with time-varying transition rates and
  multiple resources}, arXiv preprint arXiv:2206.07425,  (2022).

\bibitem{de2020social}
{\sc G.~F. de~Arruda, G.~Petri, and Y.~Moreno}, {\em Social contagion models on
  hypergraphs}, Physical Review Research, 2 (2020), p.~023032.

\bibitem{doshi2022convergence}
{\sc V.~Doshi, S.~Mallick, et~al.}, {\em Convergence of bi-virus epidemic
  models with non-linear rates on networks—a monotone dynamical systems
  approach}, IEEE/ACM Transactions on Networking,  (2022).

\bibitem{gallo1993directed}
{\sc G.~Gallo, G.~Longo, S.~Pallottino, and S.~Nguyen}, {\em Directed
  hypergraphs and applications}, Discrete applied mathematics, 42 (1993),
  pp.~177--201.

\bibitem{gallo2022synchronization}
{\sc L.~Gallo, R.~Muolo, L.~V. Gambuzza, V.~Latora, M.~Frasca, and
  T.~Carletti}, {\em Synchronization induced by directed higher-order
  interactions}, Communications Physics, 5 (2022), p.~263.

\bibitem{gracy2022analysis}
{\sc S.~Gracy, I.~C. Morarescu, V.~S. Varma, and P.~E. Pare}, {\em Analysis and
  on/off lockdown control for time-varying sis epidemics with a shared
  resource}, in 2022 European Control Conference (ECC), IEEE, 2022,
  pp.~1660--1665.

\bibitem{gracy2022endemic}
{\sc S.~Gracy, M.~Ye, B.~Anderson, and C.~A. Uribe}, {\em On the endemic
  behavior of a competitive tri-virus sis networked model}, arXiv preprint
  arXiv:2209.11826,  (2022).

\bibitem{hethcote2000mathematics}
{\sc H.~W. Hethcote}, {\em The mathematics of infectious diseases}, SIAM
  review, 42 (2000), pp.~599--653.

\bibitem{iacopini2019simplicial}
{\sc I.~Iacopini, G.~Petri, A.~Barrat, and V.~Latora}, {\em Simplicial models
  of social contagion}, Nature communications, 10 (2019), pp.~1--9.

\bibitem{illner2016sis}
{\sc R.~Illner and J.~Ma}, {\em An sis-type marketing model on random
  networks}, Communications in mathematical sciences, 14 (2016),
  pp.~1723--1740.

\bibitem{janson2020analysis}
{\sc A.~Janson, S.~Gracy, P.~E. Par{\'e}, H.~Sandberg, and K.~H. Johansson},
  {\em Analysis of a networked sis multi-virus model with a shared resource},
  IFAC-PapersOnLine, 53 (2020), pp.~797--802.

\bibitem{janson2020networked}
{\sc A.~Janson, S.~Gracy, P.~E. Par{\'e}, H.~Sandberg, and K.~H. Johansson},
  {\em Networked multi-virus spread with a shared resource: Analysis and
  mitigation strategies}, arXiv preprint arXiv:2011.07569,  (2020).

\bibitem{jardon2021geometric}
{\sc H.~Jard{\'o}n-Kojakhmetov, C.~Kuehn, A.~Pugliese, and M.~Sensi}, {\em {A
  geometric analysis of the SIR, SIRS and SIRWS epidemiological models}},
  Nonlinear Analysis: Real World Applications, 58 (2021), p.~103220.

\bibitem{kephart1992directed}
{\sc J.~O. Kephart and S.~R. White}, {\em Directed-graph epidemiological models
  of computer viruses}, in Computation: the micro and the macro view, World
  Scientific, 1992, pp.~71--102.

\bibitem{khalil2002nonlinear}
{\sc H.~K. Khalil}, {\em Nonlinear systems third edition}, Patience Hall, 115
  (2002).

\bibitem{khanafer2014stability}
{\sc A.~Khanafer, T.~Ba{\c{s}}ar, and B.~Gharesifard}, {\em Stability
  properties of infected networks with low curing rates}, in 2014 American
  Control Conference, IEEE, 2014, pp.~3579--3584.

\bibitem{kohler2021robust}
{\sc J.~K{\"o}hler, L.~Schwenkel, A.~Koch, J.~Berberich, P.~Pauli, and
  F.~Allg{\"o}wer}, {\em Robust and optimal predictive control of the covid-19
  outbreak}, Annual Reviews in Control, 51 (2021), pp.~525--539.

\bibitem{kovacevic2022distributed}
{\sc R.~M. Kovacevic, N.~I. Stilianakis, and V.~M. Veliov}, {\em A distributed
  optimal control model applied to covid-19 pandemic}, SIAM Journal on Control
  and Optimization, 60 (2022), pp.~S221--S245.

\bibitem{lavine2011natural}
{\sc J.~S. Lavine, A.~A. King, and O.~N. Bj{\o}rnstad}, {\em Natural immune
  boosting in pertussis dynamics and the potential for long-term vaccine
  failure}, Proceedings of the National Academy of Sciences, 108 (2011),
  pp.~7259--7264.

\bibitem{li2017survey}
{\sc M.~Li, X.~Wang, K.~Gao, and S.~Zhang}, {\em A survey on information
  diffusion in online social networks: Models and methods}, Information, 8
  (2017), p.~118.

\bibitem{li2022competing}
{\sc W.~Li, X.~Xue, L.~Pan, T.~Lin, and W.~Wang}, {\em Competing spreading
  dynamics in simplicial complex}, Applied Mathematics and Computation, 412
  (2022), p.~126595.

\bibitem{8954786}
{\sc F.~Liu and M.~Buss}, {\em Optimal control for heterogeneous node-based
  information epidemics over social networks}, IEEE Transactions on Control of
  Network Systems, 7 (2020), pp.~1115--1126,
  \url{https://doi.org/10.1109/TCNS.2019.2963488}.

\bibitem{liu2020stability}
{\sc F.~Liu, C.~Shaoxuan, X.~Li, and M.~Buss}, {\em On the stability of the
  endemic equilibrium of a discrete-time networked epidemic model},
  IFAC-PapersOnLine, 53 (2020), pp.~2576--2581.

\bibitem{liu2019networked}
{\sc J.~Liu, P.~E. Par{\'e}, E.~Du, and Z.~Sun}, {\em A networked sis disease
  dynamics model with a waterborne pathogen}, in 2019 American Control
  Conference (ACC), IEEE, 2019, pp.~2735--2740.

\bibitem{liu2019analysis}
{\sc J.~Liu, P.~E. Par{\'e}, A.~Nedi{\'c}, C.~Y. Tang, C.~L. Beck, and
  T.~Ba{\c{s}}ar}, {\em Analysis and control of a continuous-time bi-virus
  model}, IEEE Transactions on Automatic Control, 64 (2019), pp.~4891--4906.

\bibitem{liu1987dynamical}
{\sc W.-m. Liu, H.~W. Hethcote, and S.~A. Levin}, {\em Dynamical behavior of
  epidemiological models with nonlinear incidence rates}, Journal of
  mathematical biology, 25 (1987), pp.~359--380.

\bibitem{lloyd2001viruses}
{\sc A.~L. Lloyd and R.~M. May}, {\em How viruses spread among computers and
  people}, Science, 292 (2001), pp.~1316--1317.

\bibitem{lukacs2014probability}
{\sc E.~Lukacs}, {\em Probability and mathematical statistics: an
  introduction}, Academic Press, 2014.

\bibitem{matamalas2018effective}
{\sc J.~T. Matamalas, A.~Arenas, and S.~G{\'o}mez}, {\em Effective approach to
  epidemic containment using link equations in complex networks}, Science
  advances, 4 (2018), p.~eaau4212.

\bibitem{matamalas2020abrupt}
{\sc J.~T. Matamalas, S.~G{\'o}mez, and A.~Arenas}, {\em Abrupt phase
  transition of epidemic spreading in simplicial complexes}, Physical Review
  Research, 2 (2020), p.~012049.

\bibitem{mei2017dynamics}
{\sc W.~Mei, S.~Mohagheghi, S.~Zampieri, and F.~Bullo}, {\em On the dynamics of
  deterministic epidemic propagation over networks}, Annual Reviews in Control,
  44 (2017), pp.~116--128.

\bibitem{pare2018virus}
{\sc P.~E. Pare}, {\em Virus spread over networks: Modeling, analysis, and
  control}, PhD thesis, University of Illinois at Urbana-Champaign, 2018.

\bibitem{pare2020modeling}
{\sc P.~E. Par{\'e}, C.~L. Beck, and T.~Ba{\c{s}}ar}, {\em Modeling,
  estimation, and analysis of epidemics over networks: An overview}, Annual
  Reviews in Control, 50 (2020), pp.~345--360.

\bibitem{pare2022multi}
{\sc P.~E. Par{\'e}, A.~Janson, S.~Gracy, J.~Liu, H.~Sandberg, and K.~H.
  Johansson}, {\em Multi-layer sis model with an infrastructure network}, IEEE
  Transactions on Control of Network Systems,  (2022).

\bibitem{pare2018analysis}
{\sc P.~E. Par{\'e}, J.~Liu, C.~L. Beck, B.~E. Kirwan, and T.~Ba{\c{s}}ar},
  {\em Analysis, estimation, and validation of discrete-time epidemic
  processes}, IEEE Transactions on Control Systems Technology, 28 (2018),
  pp.~79--93.

\bibitem{pare2021multi}
{\sc P.~E. Par{\'e}, J.~Liu, C.~L. Beck, A.~Nedi{\'c}, and T.~Ba{\c{s}}ar},
  {\em Multi-competitive viruses over time-varying networks with mutations and
  human awareness}, Automatica, 123 (2021), p.~109330.

\bibitem{pare2020analysis}
{\sc P.~E. Par{\'e}, D.~Vrabac, H.~Sandberg, and K.~H. Johansson}, {\em
  Analysis, online estimation, and validation of a competing virus model}, in
  2020 American Control Conference (ACC), IEEE, 2020, pp.~2556--2561.

\bibitem{9029305}
{\sc P.~E. Paré, J.~Liu, H.~Sandberg, and K.~H. Johansson}, {\em Multi-layer
  disease spread model with a water distribution network}, in 2019 IEEE 58th
  Conference on Decision and Control (CDC), 2019, pp.~8335--8340,
  \url{https://doi.org/10.1109/CDC40024.2019.9029305}.

\bibitem{prasse2019viral}
{\sc B.~Prasse and P.~Van~Mieghem}, {\em The viral state dynamics of the
  discrete-time nimfa epidemic model}, IEEE Transactions on Network Science and
  Engineering, 7 (2019), pp.~1667--1674.

\bibitem{qi2017tensor}
{\sc L.~Qi and Z.~Luo}, {\em Tensor analysis: spectral theory and special
  tensors}, SIAM, 2017.

\bibitem{rodrigues2016can}
{\sc H.~S. Rodrigues and M.~J. Fonseca}, {\em Can information be spread as a
  virus? viral marketing as epidemiological model}, Mathematical methods in the
  applied sciences, 39 (2016), pp.~4780--4786.

\bibitem{ruan2003dynamical}
{\sc S.~Ruan and W.~Wang}, {\em Dynamical behavior of an epidemic model with a
  nonlinear incidence rate}, Journal of differential equations, 188 (2003),
  pp.~135--163.

\bibitem{sahneh2014competitive}
{\sc F.~D. Sahneh and C.~Scoglio}, {\em Competitive epidemic spreading over
  arbitrary multilayer networks}, Physical Review E, 89 (2014), p.~062817.

\bibitem{sauer2021identifiability}
{\sc T.~Sauer, T.~Berry, D.~Ebeigbe, M.~M. Norton, A.~J. Whalen, and S.~J.
  Schiff}, {\em Identifiability of infection model parameters early in an
  epidemic}, SIAM Journal on Control and Optimization, 60 (2021), pp.~S27--S48.

\bibitem{shapiro2016fixed}
{\sc J.~H. Shapiro}, {\em A fixed-point farrago}, Springer, 2016.

\bibitem{st2021universal}
{\sc G.~St-Onge, H.~Sun, A.~Allard, L.~H{\'e}bert-Dufresne, and G.~Bianconi},
  {\em Universal nonlinear infection kernel from heterogeneous exposure on
  higher-order networks}, Physical review letters, 127 (2021), p.~158301.

\bibitem{stella2022role}
{\sc L.~Stella, A.~P. Mart{\'i}nez, D.~Bauso, and P.~Colaneri}, {\em The role
  of asymptomatic infections in the covid-19 epidemic via complex networks and
  stability analysis}, SIAM Journal on Control and Optimization, 60 (2022),
  pp.~S119--S144.

\bibitem{tien2010multiple}
{\sc J.~H. Tien and D.~J. Earn}, {\em Multiple transmission pathways and
  disease dynamics in a waterborne pathogen model}, Bulletin of mathematical
  biology, 72 (2010), pp.~1506--1533.

\bibitem{trpevski2010model}
{\sc D.~Trpevski, W.~K. Tang, and L.~Kocarev}, {\em Model for rumor spreading
  over networks}, Physical Review E, 81 (2010), p.~056102.

\bibitem{van2013homogeneous}
{\sc P.~Van~Mieghem and J.~Omic}, {\em In-homogeneous virus spread in
  networks}, arXiv preprint arXiv:1306.2588,  (2013).

\bibitem{van2008virus}
{\sc P.~Van~Mieghem, J.~Omic, and R.~Kooij}, {\em Virus spread in networks},
  IEEE/ACM Transactions On Networking, 17 (2008), pp.~1--14.

\bibitem{vespignani2020modelling}
{\sc A.~Vespignani, H.~Tian, C.~Dye, J.~O. Lloyd-Smith, R.~M. Eggo,
  M.~Shrestha, S.~V. Scarpino, B.~Gutierrez, M.~U. Kraemer, J.~Wu, et~al.},
  {\em Modelling covid-19}, Nature Reviews Physics, 2 (2020), pp.~279--281.

\bibitem{yang2015impact}
{\sc L.-X. Yang and X.~Yang}, {\em The impact of nonlinear infection rate on
  the spread of computer virus}, Nonlinear Dynamics, 82 (2015), pp.~85--95.

\bibitem{yang2017bi}
{\sc L.-X. Yang, X.~Yang, and Y.~Y. Tang}, {\em A bi-virus competing spreading
  model with generic infection rates}, IEEE Transactions on Network Science and
  Engineering, 5 (2017), pp.~2--13.

\bibitem{yang2012towards}
{\sc X.~Yang and L.-X. Yang}, {\em Towards the epidemiological modeling of
  computer viruses}, Discrete Dynamics in Nature and Society, 2012 (2012).

\bibitem{ye2021parameter}
{\sc L.~Ye, P.~E. Par{\'e}, and S.~Sundaram}, {\em Parameter estimation in
  epidemic spread networks using limited measurements}, SIAM Journal on Control
  and Optimization, 60 (2021), pp.~S49--S74.

\bibitem{ye2022convergence}
{\sc M.~Ye, B.~D. Anderson, and J.~Liu}, {\em Convergence and equilibria
  analysis of a networked bivirus epidemic model}, SIAM Journal on Control and
  Optimization, 60 (2022), pp.~S323--S346.

\bibitem{yi2022edge}
{\sc Y.~Yi, L.~Shan, P.~E. Par{\'e}, and K.~H. Johansson}, {\em Edge deletion
  algorithms for minimizing spread in sir epidemic models}, SIAM Journal on
  Control and Optimization, 60 (2022), pp.~S246--S273.

\bibitem{yuan2012modeling}
{\sc H.~Yuan, G.~Liu, and G.~Chen}, {\em On modeling the crowding and
  psychological effects in network-virus prevalence with nonlinear epidemic
  model}, Applied Mathematics and Computation, 219 (2012), pp.~2387--2397.

\bibitem{zhu2014demographic}
{\sc X.~Zhu, Y.~Nie, and A.~Li}, {\em Demographic prediction of online social
  network based on epidemic model}, in Web Technologies and Applications: APWeb
  2014 Workshops, SNA, NIS, and IoTS, Changsha, China, September 5, 2014.
  Proceedings 16, Springer, 2014, pp.~93--103.

\bibitem{zino2021analysis}
{\sc L.~Zino and M.~Cao}, {\em Analysis, prediction, and control of epidemics:
  A survey from scalar to dynamic network models}, IEEE Circuits and Systems
  Magazine, 21 (2021), pp.~4--23.

\end{thebibliography}
\end{document}